\documentclass[letterpaper,twocolumn,10pt]{article}
\usepackage{usenix-2020-09}

\usepackage{amsmath}
\usepackage{amsfonts}
\usepackage{amssymb}
\usepackage{balance}
\usepackage{booktabs}
\usepackage{multirow}
\usepackage{semantic}
\usepackage{tikz}
\usepackage{xcolor,colortbl}
\usepackage{ulem}
\usepackage{amsthm}
\usepackage{esvect}
\usepackage{breakurl}
\usepackage{pifont}

\usepackage[ruled,linesnumbered,vlined]{algorithm2e}

\renewcommand{\emph}[1]{\textit{#1}}

\newtheorem{theorem}{Theorem}
\newtheorem{lemma}{Lemma}
\newtheorem{definition}{Definition}[section]

\newif\ifdiff
\difftrue
\difffalse
\ifdiff
\newcommand{\revise}[1]{\textcolor{blue}{#1}}
\newcommand{\delete}[1]{\textcolor{red}{\sout{#1}}}
\newcommand{\replace}[2]{\textcolor{red}{\sout{#1}}\textcolor{blue}{#2}}
\newcommand{\startreviseblock}{\color{blue}}
\newcommand{\dismissreviseblock}{\color{black}}
\newcommand{\shi}[1]{\textcolor{cyan}{[shi: #1]}}
\newcommand{\xz}[1]{\textcolor{purple}{[xz: #1]}}
\newcommand{\remark}[1]{\textcolor{orange}{[remark: #1]}}
\else
\newcommand{\revise}[1]{{{#1}}}
\newcommand{\delete}[1]{}
\newcommand{\replace}[2]{#2}
\newcommand{\startreviseblock}{}
\newcommand{\dismissreviseblock}{}
\newcommand{\shi}[1]{}
\newcommand{\xz}[1]{}
\newcommand{\remark}[1]{}
\fi

\newif\ifshowminordiff
\showminordifftrue
\showminordifffalse
\ifshowminordiff
\newcommand{\mrevise}[1]{\revise{#1}}
\newcommand{\mdelete}[1]{\delete{#1}}
\newcommand{\mreplace}[2]{\textcolor{red}{\sout{#1}}\textcolor{blue}{#2}}
\newcommand{\mremark}[1]{\remark{#1}}
\newcommand{\mstartreviseblock}{\startreviseblock}
\newcommand{\mdismissreviseblock}{\dismissreviseblock}
\else
\newcommand{\mrevise}[1]{#1}
\newcommand{\mdelete}[1]{}
\newcommand{\mreplace}[2]{#2}
\newcommand{\mremark}[1]{}
\newcommand{\mstartreviseblock}{}
\newcommand{\mdismissreviseblock}{}
\fi

\newcommand{\usenixtitle}[1]{\title{\Large \bf #1}}
\newcommand{\tool}{{StateLifter}}
\newcommand{\proteus}{{Proteus}}
\newcommand{\autoformat}{{AutoFormat}}
\newcommand{\tupni}{{Tupni}}
\newcommand{\reverx}{{ReverX}}
\newcommand{\nemesys}{{NemeSys}}
\newcommand{\defparbf}[1]{\smallskip\noindent\textbf{#1}}
\newcommand{\defparit}[1]{\smallskip\noindent\textit{\underline{\smash{#1}}}}
\newcommand{\bS}{\mathbb{S}}
\newcommand{\bE}{\mathbb{E}}
\newcommand{\bF}{\mathbb{F}}
\newcommand{\bI}{\mathbb{I}}

\newcommand{\bR}{\mathbb{R}}
\newcommand{\bM}{\mathbb{M}}
\newcommand{\bV}{\mathbb{V}}
\newcommand{\qcolon}{\mbox{`:'}}
\newcommand{\qcircum}{\mbox{`\textasciicircum'}}
\newcommand{\qa}{\mbox{`a'}}
\newcommand{\qb}{\mbox{`b'}}
\newcommand{\qc}{\mbox{`c'}}
\newcommand{\qz}{\mbox{`z'}}

\begin{document}
    
\date{}

\usenixtitle{Extracting Protocol Format as State Machine via Controlled Static Loop Analysis}

\author{
    {\rm Qingkai Shi}\\
    Purdue University
    \and
    {\rm Xiangzhe Xu}\\
    Purdue University
    \and
    {\rm Xiangyu Zhang}\\
    Purdue University
} 

\maketitle

\begin{abstract}

Reverse engineering of protocol message formats is critical for many security applications.
Mainstream techniques \mdelete{in this field focus on}\mrevise{use} dynamic analysis and inherit its low-coverage problem --- 
the inferred message formats only reflect the features of their inputs.
To achieve high coverage,
we choose to use static analysis to infer message formats from the implementation of protocol parsers.
In this work,
we focus on a class of extremely challenging protocols whose formats are described via 
\revise{constraint-enhanced} regular expressions and parsed using finite state machines.
Such state machines are often implemented as complicated parsing loops,
which are inherently difficult to analyze via conventional static analysis.
Our new technique extracts a state machine \mdelete{from a parsing loop }by regarding 
each loop iteration as a state and the dependency between loop iterations as state transitions.
To achieve high, i.e., path-sensitive, precision but avoid path explosion,
the analysis is controlled to merge as many paths as possible based on carefully-designed rules.
The evaluation results show that 
we can infer a state machine and, thus, the message formats,
in five minutes with over 90\% precision and recall, far better
than the state of the art\mdelete{s}.
We also applied the state machines to enhance protocol fuzzers,
which are improved by 20\% to 230\% in terms of coverage\mreplace{.
We found twelve zero-days with the guidance of inferred state machines. 
The baseline without guidance only finds two.
We also provide case studies of applying our approach to other domains beyond network protocols.}{ and detect ten more zero-days compared to baselines.}
\end{abstract}

\section{Introduction}\label{sec:introduction}

In the era of the internet of things,
any vulnerability in network protocols may lead to devastating consequences for countless devices that are inter-connected and spread worldwide.
For instance, in 2020,
a protocol vulnerability led to the largest ever DDoS attack that targeted Amazon Web Service, affecting millions of active users~\cite{aws2020ddos}.
To ensure protocol security by automated analyses including fuzzing~\cite{gascon2015pulsar,comparetti2009prospex}, model checking~\cite{musuvathi2004model,bishop2005rigorous}, verification~\cite{bishop2006engineering}, and many others,
a key prerequisite is to acquire a formal specification of the message formats.
However,
this is a hard challenge.

There have been many works on automatically inferring the formats of network messages~\cite{narayan2015survey,sija2018survey,duchene2018state,li2011survey}.
However, almost all existing works are in a fashion of dynamic analysis --- either network trace analysis~\cite{cui2007discoverer,kleber2018nemesys,kleber2020message,ye2021netplier,wang2011biprominer,wang2012semantics,luo2013position} or dynamic program analysis~\cite{gopinath2020mining,hoschele2016mining,caballero2009dispatcher,lin2008automatic,cui2008tupni,caballero2007polyglot,wang2009reformat,lin2010reverse,liu2013extracting}.
The former captures online network traces and 
uses statistical methods including machine learning to cluster the traces into different categories and then perform message alignment and field identification.
The latter runs the captured network traces against the protocol implementation
and leverages the runtime control or data flows to infer message formats.
Despite being useful in many applications,
as dynamic analyses,
they cannot infer message formats not captured by the input network traces.
For instance,
a recent work reported a highly precise technique but with
coverage lower than 0.1~\cite{ye2021netplier}.
This means that it may miss \mdelete{important }message formats \mreplace{and degrade the performance of}{that are important for} downstream security analysis.

To infer message formats with high coverage,
\mdelete{in this work,}
we use static analysis, which does not rely on any input network traces but can thoroughly analyze a protocol parser.
We target open protocols that have publicly available source code.
While these protocols often have available specifications,
they are usually documented in a natural language that is not machine-readable and contains inconsistencies, ambiguities, and even vulnerabilities~\cite{mcquistin2020parsing}.
Hence, inferring formal specifications for open protocols deserve dedicated studies.
Particularly, we target a category of extremely challenging protocols, namely \textit{regular protocols},
\revise{which have two main features.
    First, the format of a regular protocol can be specified by a constraint-enhanced regular expression (ce-regex),
    such as $(a|b)^{+}c$ where $a$, $b$, and $c$ are respectively one-, two-, and four-byte variables satisfying the constraints $a~\mbox{{mod}}~10 = 4$, $b > 3$, and $(c \gg 16) + c > 100$.
    Compared to a common regular expression (com-regex), the constraints in a ce-regex allow us to specify rich semantics in a network protocol.
    Note that a com-regex can be regarded as a simple instance of ce-regex.
    For instance, a com-regex $(a|b)^{+}c$ can be viewed as a ce-regex with the constraints
    $a = \qa$, $b = \qb$, and $c = \qc$.
    Second, the messages of a regular protocol are parsed via a finite state machine. This is common in performance-sensitive and embedded systems for the benefit of low latency~\cite{graham2014finite}. 
    That is, with a state machine, we can parse a protocol without waiting for the entire message --- whenever receiving a byte, we parse it and record the current state; the recorded state allows us to continue parsing once we receive the next byte.}
\mdelete{whose formats can be described as regular expressions
and are parsed via finite state machines for the benefit of low latency.
There are \textit{three key differences} between our approach and the state of the arts.
First,
we do not assume the availability of network traces which, however, are required by existing works as their inputs
but could be hard to obtain in many cases.
Hence, our approach could be a promising alternative especially when high-quality network traces are not available.
Second,
we provide a different perspective to understand the message formats in protocol reverse engineering.
Existing works understand message formats by segmenting a message into multiple fields while we understand message formats via its parsing state machine.
Such state machines allow us to specify message formats with both high precision and high coverage while existing dynamic-analysis-based approaches often have the coverage problem and, as will be illustrated in \S\ref{sec:motivating_example}, they are not effective when dealing with state-machine-based parsers.
Third,
our work is also different from many previous works that infer the system state machines such as the one describing TCP's handshake mechanism. In this work,
the state machine is an equivalent representation of the message formats.}\mremark{we move this part to the end of this section}

It is inherently challenging for static program analysis to infer the formats of a regular protocol from its parser.
This is because a state machine for parsing is often implemented as a multi-path loop\footnote{A single-path loop contains only a single path in its loop body. A multi-path loop contains multiple paths in its loop body.} that involves complex path interleaving that mimics the state transitions,
but conventional static analysis --- loop unwinding, loop invariant inference, and loop summarization --- cannot handle such loops well.
First,
loop unwinding unrolls a loop with a limited number of iterations and, hence, will miss program behaviors beyond the unrolling times.
Second,
loop invariant techniques compute properties that always hold in each loop iteration.
They rely on abstract interpretation for fixed point computation and, to ensure the termination, use the widening operators that often lead to significant loss of precision~\cite{mine2006octagon,kroening2013loop,ancourt2010modular,gopan2006lookahead,gopan2007guided,gupta2009invgen,sharma2011simplifying,jeannet2014abstract,nguyen2014using}.
Third,
loop summarization techniques precisely infer the input and output relations of a loop
by induction.
They are good at handling single-path loops~\cite{godefroid2011automatic,saxena2009loop} or some simple multi-path loops~\cite{strejvcek2012abstracting,xie2015slooper}.
When used to analyze a multi-path loop that implements a state machine,
they either fail to work or have to enumerate all paths in the loop body~\cite{xie2016proteus,xie2019automatic}, thus suffering from path explosion.
The path explosion problem not only significantly slows down the static analysis
but also leads to the explosion of states and state transitions, making the output state machine not operable.

\mdelete{Our static analysis extracts a state machine from a parsing loop by regarding}\mrevise{To infer state machines from a parsing loop, our static analysis regards} each loop iteration as a state and the dependency between loop iterations as state transitions.
It mitigates the path explosion problem with the key insight that
a state machine can be significantly compressed by merging states and state transitions.
For instance, both state machines in Figure~\ref{fig:insight} represent the com-regex $(a|b)^{+}c$,
but the one in Figure~\ref{fig:insight}(b) is notably compressed.
This observation guides us to design a static analysis that merges as many program paths as possible when analyzing an iteration of the parsing loop,
producing a super state for the merged program paths, e.g., the state $F$,  instead of many small states for individual program paths, e.g., the states $B$ and $C$.
As a result, our analysis notably alleviates the path explosion problem and infers highly compressed state machines, e.g., Figure~\ref{fig:insight}(b), even from the implementation of complex state machines, e.g., Figure~\ref{fig:insight}(a).
As for state transitions,
we record the pre- and post-condition of each loop iteration.
These conditions
allow us to compute~the dependency between two consecutive loop iterations
and are regarded as state-transition constraints.
As a whole,
an inferred state machine represents the message formats and can drive many security analyses.

\begin{figure}[t]
    \centering
    \includegraphics[width=0.99\columnwidth]{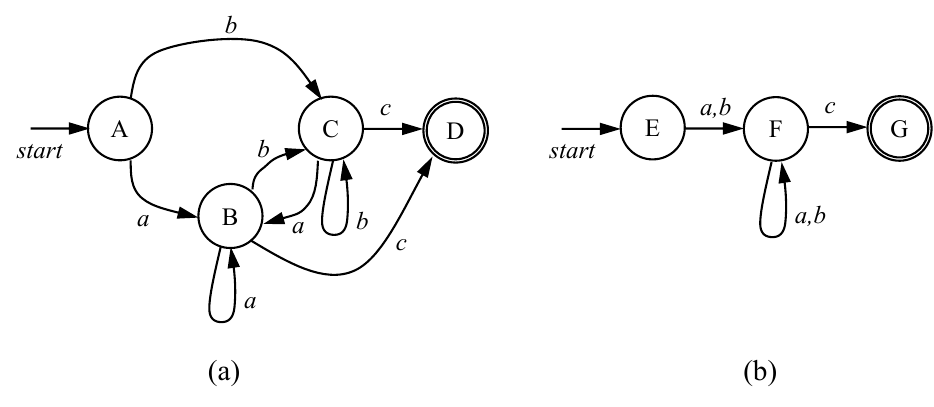}
    \caption{Example to illustrate the insight of our approach.}\label{fig:insight}
\end{figure}

\mrevise{There are three key differences between our approach and the state of the art.
First,
we do not assume the availability of network traces which, however, are required by existing works
but could be hard to obtain~\cite{narayan2015survey}.
Hence, our approach could be a promising alternative when high-quality network traces are not available.
Second,
different from many existing works that understand message formats by segmenting a message into multiple fields,
we understand message formats via the parsing state machine.
Such state machines allow us to specify message formats with both high precision and high coverage and, as will be illustrated in \S\ref{sec:motivating_example}, they are not effective when dealing with state-machine-based parsers, thus exhibiting low precision and recall.
Third,
our work is also different from many previous works~\cite{comparetti2009prospex,ye2021netplier,leita2005scriptgen,cui2006protocol,shevertalov2007reverse,antunes2011reverse,wang2011inferring,cho2010inference,zhang2012mining,laroche2012network,mcmahon2022closer,joeri2015protocol} that infer system state machines such as the one describing TCP's handshake mechanism. In this work,
state machines are used to specify message formats.}
In summary, we make \mrevise{the} following contributions.
\begin{itemize}
    \item We developed a novel static analysis that mitigates the path-explosion problem in conventional approaches and can infer highly compressed state machines from code.
    
    \item We applied the static analysis to reverse engineering message formats. The analysis is highly precise and fast with high coverage. 
    To the best of our knowledge, this is the first static analysis that formulates the problem of message format inference as extracting state machines.
    
    \item 
    We implemented our approach, namely \tool, and evaluated it on ten protocols from different domains.
    \tool\ is highly efficient as it can infer a parsing state machine or, equivalently, the message formats in five minutes.
    \tool\ is also highly precise with a high recall as its inferred state machine can uncover $\ge 90\%$ protocol formats with $\le 10\%$
    false ones. 
    By contrast, the baselines often miss $\ge 50\%$ of possible formats and may produce $\ge 40\%$ false ones.
    We use the inferred finite state machines to improve two state-of-the-art protocol fuzzers.
    The results demonstrate that, with the inferred state machines, the fuzzers can be improved by 20\% to 230\% in terms of coverage.
    We have discovered 12 zero-day vulnerabilities but the baseline fuzzers only find two of them. 
\end{itemize}

\newpage

\section{Problem Scope}\label{sec:problem_scope}

\mreplace{In this paper, w}{W}e target regular protocols,
of which (1) the message formats can be described as \revise{constraint-enhanced} regular expressions and (2) the messages are parsed via finite state machines (FSM).
\mrevise{Formally, 
    considering the equivalence of regular expression and FSM,
    we define a regular protocol in Definition~\ref{def} as an FSM enhanced by first-order logic constraints.}
The problem we address is to infer the \mreplace{parsing state machine and, hence, the message format, }{FSM} from the parser of a regular protocol.
\mdelete{A finite state machine (FSM) for parsing}\mrevise{An FSM} can be either deterministic or \mreplace{non-deterministic}{not}.
Since \mdelete{a deterministic FSM is just a special form of the non-deterministic FSM
and }any non-deterministic FSM can be converted to a deterministic \mreplace{FSM}{one},
for simplicity, FSM means non-deterministic FSM by default in this paper. Note that a non-deterministic FSM may contain multiple start states and a state may transition to multiple successor states with the same inputs.
\begin{definition}\label{def}
    An FSM is a quintuple $(\Sigma, \bS, \bS_0, \bF, \delta)$ where
    \begin{itemize}
        \item $\Sigma$ is a set of first-order logic constraints over a byte sequence $\sigma^n$ of length $n$. We use $\sigma^n_i$ and $\sigma^n_{i..j}$ to represent the $i+1$th byte and a subsequence of $\sigma^n$, respectively.
        A typical constraint could be $\sigma^2_1\sigma^2_0 > 10$, which means that the value of a two-byte integer with $\sigma^2_1$ the most significant byte and $\sigma^2_0$ the least is larger than ten. We simply write $\sigma$ as a shorthand of $\sigma^1_0$ and $\sigma^1$.
        \item $\bS$ is a non-empty set of states; $\bS_0\subseteq \bS$ is a non-empty set of start states; $\bF\subseteq \bS$ is a non-empty set of final states.
        \item $\delta: \bS\times \Sigma \mapsto 2^\bS$ is the transition function, meaning that when obtaining a byte sequence satisfying a constraint at a state, we will proceed to some possible states.
    \end{itemize}
\end{definition}
By definition,
a sequence of transitions from a start state to a final state defines a possible message format.
For instance, $\delta(A\in\bS_0, \sigma^2_1\sigma^2_0 > 10) = \{ B \}$ and $\delta(B, \sigma = 5) = \{ C\in \bF \}$
are two transitions ---
one from a start state $A$ to the state $B$ with the constraint $\sigma^2_1\sigma^2_0 > 10$
and the other from the state $B$ to a final state $C$ with the constraint $\sigma = 5$.
It implies a message format where the first two bytes satisfy $\sigma^2_1\sigma^2_0 > 10$ and the third byte must be 5.
\mreplace{The inferred }{Such an }FSM \mdelete{is equivalent to the format of a regular protocol
and }allows us to generate valid messages following the state-transition constraints.
\mdelete{Our problem scope is different from existing works in the following aspects.}

\mdelete{From the perspective of inferring message formats, existing works are dynamic analysis and aim to segment ``given network messages'' into multiple fields. 
    As acknowledged by these techniques, they only infer the formats captured by the input messages and, thus, may miss important formats.
    By contrast, we statically infer the parsing FSM,
    which aims to cover all possible formats with high precision and enable us to generate network messages with high coverage. The generated messages can drive existing works for message segmentation, field inference, and others we may not support.}
    
\mdelete{From the perspective of inferring state machines, existing works are dynamic analyses and rely on ``given network messages'' to discover system state transitions such as from a state that waits for a message to a state that has received a message. In this paper,
    we do not infer such system state machines but FSMs that represent message formats.
    Since we can generate network messages with high coverage based on the inferred FSM, after input generation, 
    we can leverage existing works to infer the system state machine.}\mremark{we remove this part because it has been discussed in the last paragraph of \S\ref{sec:introduction}. We remove it to save space.}

\smallskip
\noindent
\revise{\textbf{Why Regular Protocols?}
        In practice, the formats of a wide range of network protocols, such as HTTP and UDP, can be specified via ce-regex.
        This is acknowledged by many existing works, such as LeapFrog~\cite{doenges2022leapfrog} that verifies protocol equivalence via FSMs,
        and P4~\cite{bosshart2014p4}, a domain-specific language developed by the open networking foundation, which allows us to specify protocols via FSMs. As an example, we can specify an HTTP request using the following ce-regex:\\
        \texttt{\small Method Space URI Space Version CRLF ((General-Header | Request-Header | Entity-Header) CRLF)* CRLF Body?}, \\
        where each field, e.g., Method, satisfies certain constraints such as Method = `Get' $\lor$ Method = `Post' $\lor\cdots$.}

\revise{While a protocol that can be specified by ce-regex is unnecessary to be parsed via FSMs, an FSM parser can greatly improve the performance.
    Graham and Johnson~\cite{graham2014finite} reported that an FSM parser can achieve over an order of magnitude performance improvement, and a hand-written FSM parser could scale better than widely-used implementations such as the Nginx and Apache web servers. The key factor contributing to this improvement is that an FSM parser can parse each byte of a network message as soon as the byte is received, without having to wait for the entire message.
    As an illustration, consider the FSM parser in Figure~\ref{fig:motivating-example}(a) that parses $(a|b)^{*}c$. Each iteration of the parser processes one byte received by the function \textit{read\_next\_msg\_byte()}. The parser's state, tracked by the variable \textit{state}, allows it to continue parsing once the next byte is received. Hence, we can perform important business logic, such as preparing responses and updating system status, before a full message is received.}

\revise{Due to this performance merit, regular protocols are frequently utilized in performance-critical systems, particularly in embedded systems that cannot tolerate latency. Typical examples include Mavlink~\cite{mavlink} and MQTT~\cite{mqtt-official}, both of which are well-established in their respective fields.
    Mavlink is a standard messaging protocol for communicating with unmanned vehicles and is used in popular robotic systems such as Ardupilot~\cite{ardupilot} and PX4~\cite{px4}.
    MQTT, on the other hand, is a standard messaging protocol for the internet of things and is employed across various industries, such as automotive, manufacturing, and telecommunications, to name a few. 
    In our evaluation, we include ten regular protocols from different embedded systems and designed for edge computing, musical devices, amateur radio, and many others.}

\section{\mreplace{Motivating Example}{Limitation of Existing Works}}\label{sec:motivating_example}

\begin{figure*}
    \centering
    \includegraphics[width=0.99\textwidth]{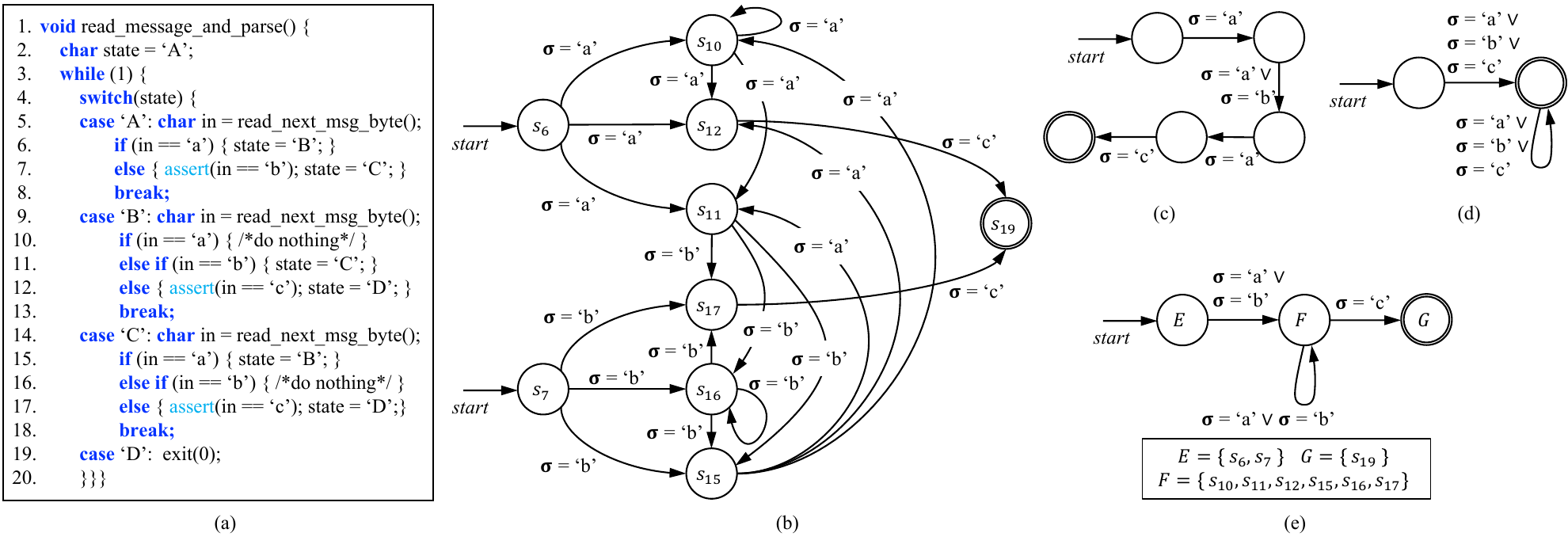}
    \ifdiff
    \else
    \vspace{-3mm}
    \fi
    \caption{(a) Implementation of the FSM in Figure~\ref{fig:insight}(a). (b) The FSM inferred by the state-of-the-art static analysis\revise{, i.e., \proteus}. (c) \revise{The FSM that represents the message format inferred by \autoformat. (d) The FSM that represents the message format inferred by \tupni. (e)} The FSM inferred by our approach, which is exactly the same as the compressed FSM in Figure~\ref{fig:insight}(b).}
    \vspace{-3mm}
    \label{fig:motivating-example}
\end{figure*}

\mdelete{Figure~\ref{fig:motivating-example} shows the implementation of the FSM in Figure~\ref{fig:insight}(a).
Let us use this example to show the problem of existing works.}

\noindent
\textbf{Network Protocol Reverse Engineering.}
Conventional techniques for inferring message formats \mreplace{of a network protocol
include}{are either} network trace analysis
\mreplace{and}{or} dynamic program analysis.
\mreplace{Given a set of messages, they focus on segmenting a message into multiple fields
and only capture the features in the input messages.
Thus, it is hard for them to infer a regular expression or an FSM to represent the message format.}{They only capture the features in a set of input messages and cannot effectively infer message formats for regular protocols.}

\defparit{(1) Network Trace Analysis (NTA).}
NTA does not analyze the implementation of protocols~\cite{cui2007discoverer,kleber2018nemesys,kleber2020message,ye2021netplier,wang2011biprominer,wang2012semantics,luo2013position,antunes2011reverse}.
Given a set of messages, they use statistical methods \replace{to cluster the messages into different categories and then perform message alignment and field identification. Hence, t}{including machine learning to identify fields in a message or infer an FSM to represent message formats. T}he formats inferred by them strongly depend on the shape of input messages.
For instance, assume that a valid message format satisfies the regular expression $(a|b)^{+}c$,
meaning that a message can start with any combination of `$a$' and `$b$'.
If all messages input to \revise{a typical} NTA\revise{, such as \reverx~\cite{antunes2011reverse} and \nemesys~\cite{kleber2018nemesys,kleber2020message},} start with `$aaa$', 
it is very likely \mdelete{for NTA }to infer an incorrect format \mreplace{that starts}{starting} with `$aaa$'\mdelete{, far from the correct format $(a|b)^{+}c$}.
\revise{In more complex cases where the format is a ce-regex, NTA cannot precisely infer  constraints in the ce-regex, e.g., $a~\mbox{{mod}}~10 = 4$, $b > 3$, and $(c \gg 16) + c > 100$. This motivates us to use program analysis so that we can precisely infer the constraints by tracking path conditions.}

\ifdiff
\else
\newpage
\fi

\defparit{(2) Dynamic Program Analysis (DPA).}
DPA is more precise than NTA as it tracks data flows in protocols' implementation~\cite{gopinath2020mining,hoschele2016mining,caballero2009dispatcher,lin2008automatic,cui2008tupni,caballero2007polyglot,wang2009reformat,lin2010reverse,liu2013extracting}.
However, it shares the same limitation with NTA as the inferred formats also only capture the features of input messages.
Typically,
\revise{techniques like \autoformat~\cite{lin2008automatic} 
infer neither repetitive fields nor field constraints.
For instance, 
given a set of messages,
e.g., $\{$ `aaac', `abac', ...  $\}$, which satisfy the ce-regex $(a|b)^{+}c$ where $a = \qa$, $b = \qb$, and $c \ge \qc$,
while \autoformat\ will run these messages against the protocol's implementation,
it does not extract conditions like $c \ge \qc$ from the code and may produce a com-regex $a(a|b)ac$ as the format.
    The FSM of the com-regex is shown in Figure~\ref{fig:motivating-example}(c),
    which is not correct as it cannot parse messages with repetitive fields and
    the last transition is not labeled by the correct constraint $\sigma\ge\qc$ and, thus, is considered to be a false transition.}

\revise{Compared to \autoformat, }\tupni~\cite{cui2008tupni} handles parsing loops with the assumption that loops are used to parse repetitive fields in a network message.
However,
this is not true for regular protocols\mrevise{. For example, Figure~\ref{fig:motivating-example} shows the implementation of the FSM in Figure~\ref{fig:insight}(a).}
\mreplace{as we}{We} can observe that \mreplace{an FSM parsing}{the} loop\mdelete{, e.g., the code in Figure~\ref{fig:motivating-example},} parses all fields in a message, no matter a field is repetitive, e.g., $a$ and $b$, or just a single byte, e.g., $c$.
Hence, \mdelete{for a given message like `$aabaaabbaaac$', }\tupni\ will produce a format like $(a|b|c)^{+}$ as the byte $c$ is also handled in the loop and regarded as a repetitive field.
\revise{Figure~\ref{fig:motivating-example}(d) shows the corresponding FSM,
    which does not represent a correct format. For example, in the inferred FSM, the incoming transitions of the final state may have the constraint $\sigma = \qa$, but in the correct FSM shown in Figure~\ref{fig:insight},
    the incoming transitions of the final state are only constrained by $\sigma = \qc$.}

\defparbf{Static Loop Analysis.}
Unlike NTA and DPA which only~capture formats in their input messages,
we propose to use static analysis to infer all possible formats in the form of FSM.
However, we fail to find any practical static analysis that can infer such formats
with high precision, recall, and speed.

\defparit{(1) Loop Unwinding and Loop Invariant.}
Loop unwinding limits the number of loop iterations to a constant $k$~\cite{xie2005scalable,shi2018pinpoint,shi2021path,babic2008calysto}.
When analyzing the parser in Figure~\ref{fig:motivating-example}(a),
it will only produce the formats of the first $k$ bytes as each iteration analyzes one byte.
Loop invariant techniques~\cite{mine2006octagon,kroening2013loop,ancourt2010modular,gopan2006lookahead,gopan2007guided,gupta2009invgen,sharma2011simplifying,jeannet2014abstract,nguyen2014using} do not infer FSMs, either. 
They compute constraints that always hold after every loop iteration.
For instance, a possible invariant of the loop in Figure~\ref{fig:motivating-example}(a) could be
$\qa < \textit{in} < \qc$. This is far from our goal of FSM inference.

\begin{figure*}
    \centering
    \includegraphics[width=0.98\textwidth]{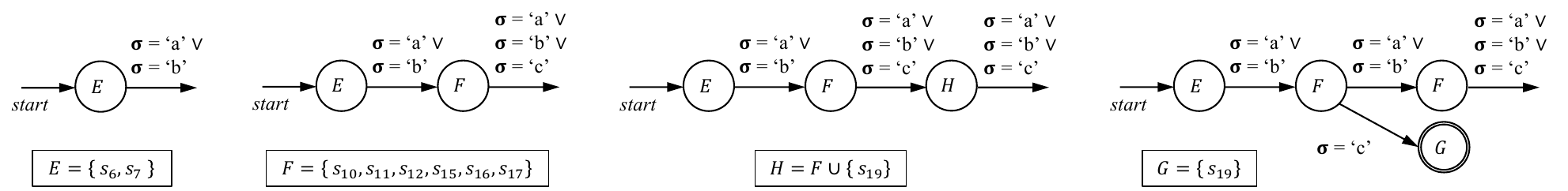}
    \vspace{-2mm}
    \caption{Basic steps of our approach.}
    \ifdiff
    \else
    \vspace{-3mm}
    \fi
    \label{fig:overview-steps}
\end{figure*}

\defparit{(2) Loop Summarization for FSM Inference.}
There are some static analyses that infer an FSM from loops~\cite{xie2016proteus,xie2019automatic,chen2019automated,shimizu2019symbolic}.
Chen et al.~\cite{chen2019automated} \mreplace{proposed a pattern-driven approach where}{assume that}
an FSM parsing loop \mreplace{is assumed to follow}{follows} a simple pattern\mreplace{.
On one hand, this approach}{ and thus} is not practical \mreplace{as}{for} real-world protocol parsers\mdelete{ rarely follow the patterns}.
For instance,
they regard a program variable as a state variable iff it is both modified in a loop iteration and referenced in future iterations.
They assume such state variables have a limited number of values, e.g., the variable \textit{state} in Figure~\ref{fig:motivating-example}
only has four possible values.
This assumption is often violated in a real protocol parser. 
A typical example is in Figure~\ref{fig:motivating-example-2} where the variable \textit{tok} satisfies their definition of state variables but its value is not enumerable.
\mreplace{On the other hand}{In addition}, this approach suffers from two explosion problems.
\mreplace{To infer FSM states}{First}, they regard every possible combination of the state variables as a state\mreplace{.
However,}{, but} the number of combinations could be explosive. For instance, if we have five state variables and each has five possible values, the resulting FSM will contain $5^5>3000$ states.
\mreplace{To infer state transitions}{Second},
they depend on symbolic execution, which is well-known to suffer from path explosion.
These explosion problems not only make static analysis unscalable \mreplace{and}{but also} significantly blow up FSMs with unnecessary states and transitions. Shimizu et al's approach has similar problems\mdelete{ as it also assumes simple code patterns and relies on symbolic execution}~\cite{shimizu2019symbolic}.

To the best of our knowledge,
\proteus~\cite{xie2016proteus,xie2019automatic} is the most recent and systematic approach to FSM inference.
It regards every path within the body of a parsing loop as an FSM state
and the dependency between two paths executed in two consecutive loop iterations as a state transition.
Figure~\ref{fig:motivating-example}(b) shows the FSM inferred by \proteus,
where $s_{i}$ represents a state and also a path that goes through Line $i$.
Each transition from $s_{i}$ to $s_{j}$ is labeled by the path condition of $s_{i}$.
It means that if the parser executes the path $s_{i}$ with its path condition,
the next iteration may execute the path $s_{j}$.
For instance,
the state transitions from $s_6$ to $s_{10}$, $s_{11}$, and $s_{12}$ are labeled by the condition $\sigma=\qa$.
It means that
if the loop executes the path $s_6$, of which the path condition is $\sigma=\qa$,
the loop may execute the path $s_{10}$, $s_{11}$, or $s_{12}$ in the next iteration.

The FSM inferred by \proteus\ is non-deterministic but \mdelete{absolutely }correct to represent the \mdelete{regular }format $(a|b)^{+}c$.
For instance, the string `$abbc$' can be parsed via the \mdelete{state }transitions $s_{6}s_{11}s_{16}s_{17}s_{19}$.
However, the FSM is too complex compared to the one we intend to implement, i.e., Figure~\ref{fig:insight}(a).
We observe that the core problem is that
it enumerates all paths in the loop body as a priori
but the number of paths \mdelete{in a program }is notoriously explosive.
Thus, the resulting FSM contains an \mreplace{explosive}{overwhelming} number of states and transitions, and \mreplace{the static analysis}{\proteus} is impractical due to path explosion.

\section{Technical Overview}\label{sec:technical_overview}

At a high level, we follow a similar idea in terms of regarding a loop iteration as an FSM state and dependency between loop iterations as state transitions.
However,
unlike \proteus, we do not enumerate all individual paths in the loop but put as many paths as possible into a path set which, as a whole, is regarded as a single FSM state.
This design simplifies the output FSM, significantly mitigates path explosion, but incurs new challenges.
In what follows, we discuss two examples,
one for our basic idea and the other for the detailed designs.

\defparbf{Basic Idea: Path Set as State.}
\mreplace{Our approach performs}{We perform} a precise abstract interpretation over each iteration of the parsing loop.
The basic steps of analyzing the code in Figure~\ref{fig:motivating-example} are shown in Figure~\ref{fig:overview-steps}.
In the first iteration of the parsing loop, due to the initial value of the variable \textit{state},
we analyze the paths $s_6$ and $s_7$,
depending on the condition: $\Phi_E\equiv \sigma=\qa \lor \sigma=\qb$.
Thus,
we create the state $E$ to represent the path set $\{s_6, s_7\}$
and label the outgoing edge of $E$ with the condition $\Phi_E$.

After the first iteration, the value of the variable \textit{state} is either `B‘ or `C'.
Thus, in the second iteration,
the abstract interpretation analyzes all paths in $F = \{ s_{10}, s_{11}, s_{12}, s_{15}, s_{16}, s_{17} \}$
with the path condition $\Phi_F\equiv \sigma=\qa \lor \sigma=\qb \lor \sigma=\qc$.
Hence, we create the state $F$ with the outgoing condition $\Phi_F$.

After the second iteration,
the value of the variable \textit{state} could be `B‘, `C', or `D'.
Thus, in the third iteration,
we analyze the paths in $H = F\cup G, G=\{s_{19}\}$
with the path condition $\Phi_H\equiv \sigma=\qa \lor \sigma=\qb \lor \sigma=\qc$.
Hence, we create the state $H$ with the outgoing condition $\Phi_H$.

Since the state $H$ overlaps the state $F$, we split $H$ into $F$ and $G$, just as in the last graph in Figure~\ref{fig:overview-steps}.
Since the state $H$ is split, the original edge from $F$ to $H$ is also split accordingly.
For instance, the condition from $F$ to $G$ is $\sigma=\qc$ because, only when we go through the paths $s_{12}, s_{17}\in F$, of which the path condition is $\sigma=\qc$,
we can reach the path $s_{19}\in G$.
The state $G$ is a final state because it stands for the path $s_{19}$ that leaves the parsing loop.
Finally, we merge the two $F$ states, forming a self-cycle as \mreplace{in the final FSM}{illustrated} in Figure~\ref{fig:motivating-example}\replace{(c)}{(e)}.

\newcommand{\myalg}[1]{
\begin{algorithm}[t]\footnotesize
    \caption{State Machine Inference.} #1
    \SetKwFunction{InferFSM}{infer\_state\_machine}
    \SetKwProg{Proc}{Procedure}{}{}
    \Proc{\InferFSM{$\mathbb{E}_\textit{init}$}}{
        $(S, \mathbb{E}_S) = \textup{\textcolor{red}{abstract\_interpretation}}~(\bE_\textit{init})$\;
        $\textit{Worklist} = \{ (S, \mathbb{E}_S) \}$;\enskip$\textit{FSM} = \emptyset$\;
        \While{\textit{Worklist} not empty}{
            $(S, \bE_S)$ = \textit{Worklist}.pop()\;
            $(S', \bE_{S'}) = \textup{\textcolor{red}{abstract\_interpretation}}~(\bE_S)$\;
            add $(S, \bE_S, S')$ into \textit{FSM}\;
            
            \color{blue}\tcc{splitting operations}\color{black}
            \ForEach{state $X$ that should be split}{
                \textup{\textcolor{red}{split}} $X$ into $X_1, X_2, \dots$\;
                replace $(X, \bE_X, Y) \in \textit{FSM}$ with $(X_i, \bE_{X_i}, Y)$\;
                replace $(Y, \bE_Y, X) \in \textit{FSM}$ with $(Y, \bE_{Y}, X_i)$\;
            }
            \textbf{assume} $S'$ is split into $S'_i$, or $S'\equiv S'_i$ if $S'$ is not split\;
            \If{$\not\exists(S_i', \bE_{S'_i}, *)\in \textit{FSM}$, where $*$ means any state}{
                add $(S_i', \bE_{S'_i})$ into \textit{Worklist}\;
            }
        
            \color{blue}\tcc{merging operations}\color{black}
            \textup{\textcolor{red}{merge}} states that represent the same path set into one state\;
            \ForEach{pair of states $(X, Y)$ such that there are multiple transitions $(X, \bE_{X1}, Y), (X, \bE_{X2}, Y), \dots \in \textit{FSM}$}{
                $\bE_X = \textup{\textcolor{red}{merge}}(\bE_{X1}, \bE_{X2}, \dots)$\;
                replace all $(X, \bE_{Xi}, Y)$ with $(X, \bE_X, Y)$ in \textit{FSM}\;
                \textbf{if~}$\forall \bE_{Xi}.\bE_{X} \not\equiv \bE_{Xi}$ \textbf{then} add $(X, \bE_X)$ into \textit{Worklist}\;
            }
        }
        \Return \textit{FSM}\;
    }
\end{algorithm}
}
\myalg{\label{alg:framework}}

\defparbf{Algorithm Framework.}\remark{Improve to a worklist algorithm for fixed-point computation. This improvement does not change our core idea but add details to facilitate the proof of convergence and soundness.}
\startreviseblock
Algorithm~\ref{alg:framework} sketches out our approach.
Its parameter is the initial program environment $\bE_\textit{init}$,
which provides necessary program information such as the initial path condition and the initial value of every program variable before entering a parsing loop.
Line~2 analyzes the first iteration of the parsing loop and outputs the analyzed path set as well as the resulting program environment, i.e., $(S, \bE_S)$.
Line~3 initializes the FSM and a worklist.

The FSM is represented by a set of state transitions. Each transition is a triple $(S, \bE_S, S')$
and describes the analyses of two consecutive iterations of the parsing loop --- 
one analyzes the path set $S$ and outputs $\bE_S$; the other uses $\bE_S$ as the precondition, which lets us analyze the path set $S'$.
Each item in the worklist is the analysis result from an iteration of the parsing loop, i.e., $(S, \bE_S)$.
We use the worklist to perform a fixed-point computation.
That is, whenever we get a new pair $(S, \bE_S)$ that has not been included in the FSM,
we add it to the worklist, because using a new $\bE_S$ as the initial program environment may result in new analysis results from the parsing loop.

Lines~5-7 continue the analysis of the next loop iteration and add the new state transition to the FSM.
Lines 8-11 split a state into multiple sub-states, just like we split the state $H$ in Figure~\ref{fig:overview-steps}.
Lines~12-14 update the worklist by adding $(S'_i, \bE_{S'_i})$ if the pair has not been included in the FSM.
Line~15 merges the states that represent the same path set, just like that we merge the two states $F$ in the last example.
If the procedure above yields multiple but non-equivalent transitions
between a pair of states, e.g., $(X, \bE_{Xi\ge1}, Y)$,
Lines~16-19 merge them into one, $(X,\bE_X, Y)$.
If $\bE_X\equiv\bE_{Xi}$,
we do not need to add $(X,\bE_{X})$ to the worklist,
because the resulting transition $(X,\bE_{Xi}, Y)$ has been in the FSM.
Otherwise, $(X,\bE_{X})$ should be added to the worklist for further computation.

The details of the merging operation will be discussed later in \S\ref{sec:approach},
but it is sound and also guarantees the convergence of a fixed-point computation. 
That is,
while we keep merging transitions from $X$ to $Y$ whenever a new transition between the two states is produced, the merging operation ensures that we will not endlessly generate new transitions from $X$ to $Y$. Instead, it will converge, i.e., reach a fixed point.
\dismissreviseblock
\delete{Algorithm~\ref{alg:framework} sketches out our approach.
Its parameter is the initial program environment $\mathbb{E}_\textit{init}$,
which provides necessary program information such as path conditions and possible values of a program variable.
Line 2 analyzes the first loop iteration and Line 3 initializes the state machine. The analysis of the first iteration returns the first state transition
where $S$ is the source state, $\top$ denotes an unknown target state, 
and the program environment $\mathbb{E}$ after the first loop iteration.
The program environment $\mathbb{E}$ includes the path condition $\Phi_S$
that we put over the edge between $S$ and $\top$.
At Lines 4-5, for every state transition that has an unknown target state,
we continue the analysis over the next loop iteration and generated a new state $S'$ as the target state.
Lines 6-7 update the FSM by adding the newly-created state transitions.
Lines 8-9 try to split the state $S'$ into multiple smaller states $S'_i$, just like we split the state $H$ in the previous example.
At the same time, the incoming program environment $\mathbb{E}$ is split into multiple incoming environments $\mathbb{E}_{S\rightarrow S_i'}$ accordingly.
This is just like when we split the state $H$, we also split the incoming condition of the state $H$ in the last example.
Lines 10-12 update the state machine by adding the new state transitions after splitting the state.
Lines 13-14 try to merge each new state with an old one,
just like we merge the two states $F$ in the last example.}

\defparbf{Controlled State Splitting and Merging.}
The previous example shows the power of regarding multiple paths as a single state,
which mitigates the path explosion problem and produces compressed FSMs.
However, \mdelete{it also introduces new challenges as }we observe that we cannot 
arbitrarily put all possible paths in a \mrevise{single} state. 
Otherwise, invalid FSMs may be generated or the algorithm performance may be seriously degraded.
Thus,
we establish dedicated rules to control state splitting and merging.
They are implemented into two key operations in Algorithm~\ref{alg:framework}, namely \texttt{split} and \texttt{merge}. 
Next, we informally discuss them in three parts: (1) we list the rules of splitting and merging states; (2) we use a detailed example to show how these rules are used; and (3) we briefly justify the rationale behind the rules.

\defparit{(1) Splitting and Merging Rules.}
We establish the following rules to split a state or merge multiple states.
\begin{itemize}
    \item Splitting Rule (\textbf{SR1}): If two states represent overlapping path sets, we split them into multiple disjoint path sets.
    This rule has been illustrated in Figure~\ref{fig:overview-steps} where the state $H$ is split into $F$ and $G$,
    so that we can reuse the state $F$.
    
    \item Splitting Rule (\textbf{SR2}): If a state represents a path set that includes both loop-exiting paths and paths that go back to the loop entry, we split it into a final state containing the exiting paths and a state containing the others. 
    Otherwise, it will be hard to decide if an FSM terminates.
    
    \item Splitting Rule (\textbf{SR3}): If a state represents a path set 
    where a variable is defined recursively in some paths, these paths should be isolated from others.
    For example, the paths $s_{12}$ and $s_{13}$ in Figure~\ref{fig:motivating-example-2}
    define the variable \textit{tok} in two manners. The path $s_{13}$ defines the variable \textit{tok} recursively based on its previous value. Hence, \mrevise{we put} the two paths $s_{12}$ and $s_{13}$ \mreplace{should not be in the same path set}{in different path sets}. 
    
\begin{figure*}[t]
    \centering
    \includegraphics[width=0.98\textwidth]{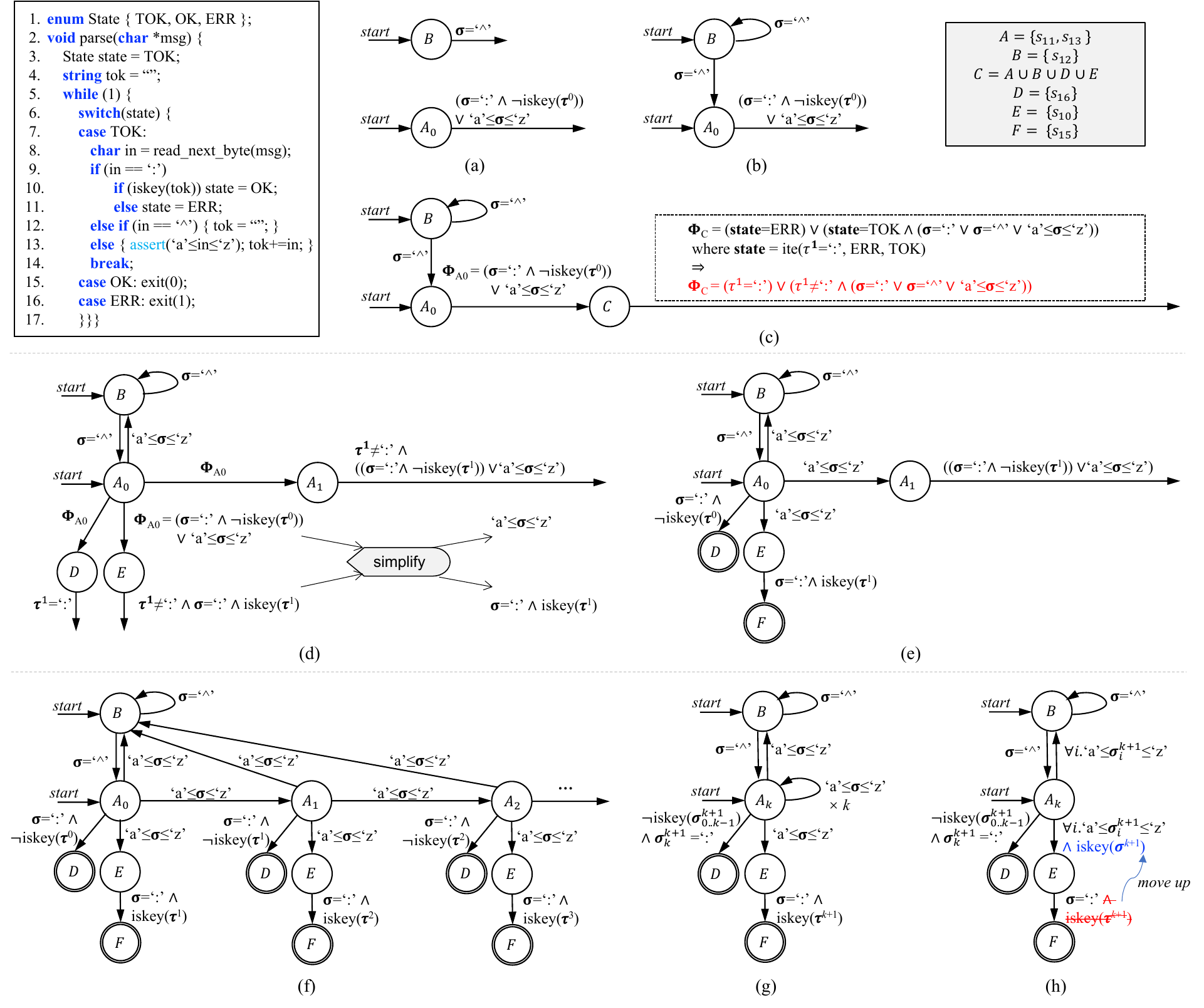}
    \vspace{-3.5mm}
    \caption{A detailed example. (a)-(h) The steps of FSM inference.}
    \vspace{-2.8mm}
    \label{fig:motivating-example-2}
\end{figure*}

    \item Merging Rule (\textbf{MR1}): 
    \mreplace{If two states represent the same path set with the same pre-/post-conditions,}{Given a set of states that represent the same path set with the same path conditions,} we merge them into a single state.
    This rule has been illustrated in Figure~\ref{fig:overview-steps} where we merge the two states $F$.
    
    \item Merging Rule (\textbf{MR2}):\remark{The change here is to include merging operations by interval-domain abstract interpretation. This allows our analysis to fall back to interval-domain abstract interpretation when induction fails, such that convergence and soundness are guaranteed.}\delete{ Given a sequence of state transitions, e.g., $\delta(A_0, \Phi_{A_0}) = \{ A_1 \}$, $\delta(A_1, \Phi_{A_1}) = \{ A_2 \}$, ..., $\delta(A_{n-1}, \Phi_{A_{n-1}}) = \{ A_{n} \}$, $\delta(A_n, \Phi_{A_n}) = \{ B \}$,
        where the states $A_i$ represent the same path set and their pre-/post-conditions aren't equivalent but inductively summarizable, e.g., $\Phi_{A_i} \equiv \sigma=i$,
        we merge them into $\delta(A, \Phi_{A}) = \{ B \}$ where $\Phi_{A}$ is the inductive form of all $\Phi_{A_i}$, e.g., $\Phi_{A}\equiv \sigma^{n+1}_k = k$.} \revise{Given a sequence of transitions between a pair of states,
    we merge them into a single transition either by induction or, if induction fails, via a widening operator from classic abstract interpretation. Let us use the following examples to illustrate.}
\startreviseblock
    \begin{itemize}
        \item Given multiple transitions between a pair of states where the transition constraints form 
        a sequence such as $\sigma=1$, $\sigma=2$, $\sigma=3$, $\dots$,
        we can apply inductive inference~\cite{angluin1983inductive} to
        merge them into a single state transition with the constraint $\sigma = k$, meaning the $k$th transition constraint.
       
        \item If the transition constraints are $\sigma=0$, $\sigma=3$, $\sigma=1$, $\dots$,
        we cannot inductively merge them as before.
        Instead, we merge them into $0\le \sigma \le 3$ using the classic widening operator from interval-domain abstract interpretation~\cite{cousot1977abstract}. This merging operation is sound but may lose precision.
    \end{itemize}
\dismissreviseblock
    
    \item Merging Rule (\textbf{MR3}): 
    \mrevise{To ensure the validity, i.e., a state transition does not refer to inputs consumed by previous transitions, we perform this rule after Algorithm~\ref{alg:framework} terminates. That is, g}\mdelete{G}iven two consecutive transitions, 
    e.g., $\delta(A, \Phi_A) = \{ B \}$ and $\delta(B, \Phi_B) = \{ C \}$,
    they are valid by definition iff $\Phi_A$ and $\Phi_B$ respectively constrain
    two consecutive but disjoint parts of an input stream.
    If the inputs constrained by $\Phi_A$ and $\Phi_B$ overlap,
    we \mrevise{either (1) replace
        the transition constraints with $\Phi_A'$ and $\Phi_B'$ such that $\Phi_A' \land \Phi_B' \equiv \Phi_A \land \Phi_B$ and
        neither $\Phi_A'$ nor $\Phi_B'$ refers to previous inputs, or (2)}
    merge the transitions, yielding $\delta(A, \Phi_A\land \Phi_B) = \{ C \}$ if $\Phi_A'$ and $\Phi_B'$ cannot be computed.
\end{itemize}

\startreviseblock
\begin{theorem}[Convergence]\label{th:convergence}
    The splitting and merging rules guarantee the convergence of Algorithm~\ref{alg:framework}.
\end{theorem}
\begin{proof}   
Given a parsing loop that contains $n$ program paths in the loop body,
SR1 ensures that we split these paths into at most $n$ disjoint path sets.
Thus, Algorithm~\ref{alg:framework} generates at most $n$ states.
While we may generate different transitions between a pair of states,
Algorithm~\ref{alg:framework} leverages MR1-2 to merge them by conventional inductive inference~\cite{angluin1983inductive} or interval-domain abstract interpretation~\cite{cousot1977abstract}, until a fixed point is reached.
Thus, we compute at most one fixed-point transition between each pair of states.
Since both the inductive inference and abstract interpretation converge,
Algorithm~\ref{alg:framework} converges after generating at most $n$ states and $n^2$ fixed-point state transitions.
\end{proof}
\dismissreviseblock

\noindent
\textit{\underline{\smash{(2) Detailed Example.}}}
\mdelete{Let us consider the code in }Figure~\ref{fig:motivating-example-2}\mdelete{,
which} shows a common \mreplace{and}{but} complex case in protocol parsers.
\mreplace{The code tries to find}{It looks for} a nonempty token \mreplace{before the delimiter \qcolon\
but after the symbol \qcircum}{between the symbol \qcircum\ and the symbol \qcolon}.
The token \textit{tok} is initialized to be an empty string and is reset
when the input is \qcircum\ (Line~12).
If the input character is a letter, 
the character is appended to \textit{tok} (Line~13).
If the input character is \qcolon, it will check if the token \textit{tok} is a nonempty keyword (Line~10).

\textbf{Figure~\ref{fig:motivating-example-2}(a).}
Since \mrevise{the variable \textit{state} and} the variable \textit{tok} \mreplace{is}{are respectively} initialized as \mrevise{\textit{TOK} and} an empty string,
in the first iteration,
we analyze the paths $s_{11}$, $s_{12}$, and $s_{13}$ as other paths are infeasible.
By SR3, the paths $s_{12}$ and $s_{13}$ cannot be in the same state.
Thus, we create the states $A_0=\{s_{11}, s_{13}\}$
and $B=\{ s_{12} \}$.
The outgoing constraint of each state is the path constraint, where we use the symbol $\tau^n$ to represent the input byte stream of length $n$ before the current loop iteration.
In the first iteration, \textit{tok} is an empty string and denoted as $\tau^0$.

\textbf{Figure~\ref{fig:motivating-example-2}(b).}
The first iteration creates two states, $A_0=\{ s_{11}, s_{13} \}$ and $B=\{ s_{12} \}$.
If we follow the state $B$, i.e., 
the first iteration runs the path $s_{12}$,
the code only resets the variable \textit{tok} and, after the reset, it is like we never enter the loop.
Hence, in the second iteration, we analyze the paths in $A_0\cup B$ again just as in the first iteration.
By MR1,
we reuse the state $A_0$ and the state $B$.
That is,
we add a self-cycle on the state $B$ and a transition from the state $B$ to the state $A_0$.

\textbf{Figure~\ref{fig:motivating-example-2}(c).}
If we follow the state $A_0$, i.e.,
the first iteration runs the paths in $A_0 = \{s_{11}, s_{13}\}$,
the second iteration will analyze the paths in $C = \{s_{10},s_{11},s_{12},s_{13},s_{16}\}$.
Thus, we create the state $C$
and add the transition from $A_0$ to $C$.
The outgoing transition of $C$ is the path condition of all paths in $C$.

\textbf{Figure~\ref{fig:motivating-example-2}(d).}
By SR1 and SR2, we split the state $C$ \mrevise{in}to four sub-states $A_1$, $B$, $D=\{s_{16}\}$, and $E=\{s_{10}\}$.
We reuse the state $B$ but create a new state $A_1$ because 
the states $A_0$ and $A_1$ have different post-conditions.
We then replace the state $C$ with the four sub-states.
The transition constraint from the state $A_0$ to each sub-state
is the original constraint from the state $A_0$ to the state $C$.
The outgoing constraint of each sub-state is the constraint of paths represented by the sub-state.
For instance, for the sub-state $D=\{s_{16}\}$, its path condition is $\textit{state}=\textup{ERR}$ where
the value of \textit{state} is $\textit{ite}(\tau^1=\qcolon, \textup{ERR}, \textup{TOK})$,
meaning that if the previous input is \qcolon,  $\textit{state}=\textup{ERR}$\mreplace{. O}{~and, o}therwise, $\textit{state}=\textup{TOK}$.
Thus, the outgoing constraint of $D$ is $\tau^1=\qcolon$.

The incoming and outgoing constraints of a state can be cross-simplified.
For instance, the outgoing constraint of $E$ includes $\tau^1\ne \qcolon$. This means that the 
incoming constraint of $E$ satisfies $\sigma\ne\qcolon$, and thus, can be simplified to $\qa\le\sigma\le\qz$.

\textbf{Figure~\ref{fig:motivating-example-2}(e).}
\mreplace{After cross-simplifying incoming/outgoing constraints,
w}{W}e continue a similar analysis of the next iteration from the states $D$, $E$, or $A_1$ because they have undetermined target states.
From the state $D=\{s_{16}\}$, since the path $s_{16}$ exits the loop, we stop the analysis and mark the state $D$ as a final state.
Similarly, we can find the final state $F$. 

\textbf{Figure~\ref{fig:motivating-example-2}(f) and Figure~\ref{fig:motivating-example-2}(g).}
If we continue the analysis from the state $A_1$,
we will find a repetitive state sequence, i.e., $A_0$, $A_1$, $A_2$, and so on.
We use MR2 to inductively merge them into $A_k$ as shown in Figure~\ref{fig:motivating-example-2}(g).
The merged state $A_k$ means the $(k+1)$th state $A$. Thus, the self-cycle on $A_k$ loops $k$ times and each time consumes an input satisfying $\qa\le\sigma\le\qb$.
For state transitions, e.g. the one from $E$ to $F$,
since the constraints between them in Figure~\ref{fig:motivating-example-2}(f) form the sequence:
$\sigma=\mbox{`:'}\land \textup{iskey}(\tau^1)$, $\sigma=\mbox{`:'}\land \textup{iskey}(\tau^2)$, $\sigma=\mbox{`:'}\land \textup{iskey}(\tau^3)$ and so on,
the transition constraint from $E$ to $F$ in Figure~\ref{fig:motivating-example-2}(g) is summarized as $\sigma=\mbox{`:'}\land \textup{iskey}(\tau^{k+1})$.

\textbf{Figure~\ref{fig:motivating-example-2}(h).}
To ensure that a state transition does not refer to symbols in previous transitions,
we merge the incoming and outgoing constraints of the state $A_k$ and $E$ by MR3, yielding the final FSM in Figure~\ref{fig:motivating-example-2}(h).
The inferred FSM is correct.
For instance,
given a string ``\textasciicircum\textasciicircum\textasciicircum abcd:'' where we assume ``abcd'' is a keyword,
the FSM can parse it by the transitions $BBBA_kEF$.
That is,
the transitions $BBBA_k$ consumes the prefix ``\textasciicircum\textasciicircum\textasciicircum''
and the transition from $A_k$ to $E$ consumes the keyword ``abcd'' by instantiating the induction variable $k=4$. \mrevise{Finally, the transition from $E$ to $F$ consumes the colon.}

\begin{figure}[t]
    \centering
    \includegraphics[width=0.94\columnwidth]{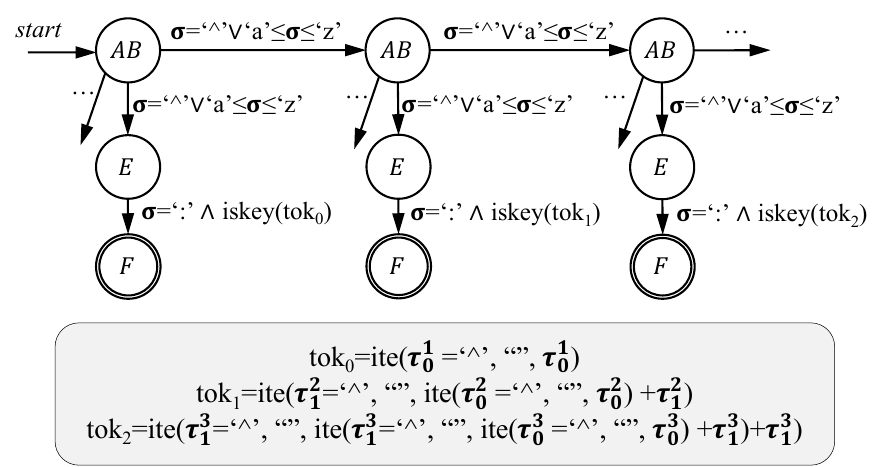}
    \vspace{-2mm}
    \caption{Violation of SR3.}\label{fig:sr3}
    \vspace{-5mm}
\end{figure}

\defparit{(3) Consequences of Violating Rules.}
\mrevise{As stated in the proof of Theorem~\ref{th:convergence},}
SR1, MR1, and MR2 contribute to the convergence of the algorithm.\mdelete{SR1 splits overlapped states so that we can reuse their common part.
MR1 and MR2 merge states to avoid repetitively generating similar states.}
Violating these rules may \mdelete{slow down the algorithm or even }make the algorithm not terminating.
SR2 and MR3 ensure the validity of an FSM by definition.
\mrevise{That is, SR2 distinguishes final states from other states, and MR3 ensures that a state transition does not refer to symbols in previous transitions.}

\mrevise{Particularly,} SR3 facilitates the use of \mrevise{induction in} MR2\mdelete{ for merging states by induction}.
Figure~\ref{fig:sr3} shows the case where we do not use SR3 and, thus, merge the states $A$
and $B$.
In this case,
after each iteration, the variable \textit{tok} may be either reset or recursively defined,
depending on if the previous input is `\textasciicircum'.
In result,
the value sequence of the variable \textit{tok}, as shown in Figure~\ref{fig:sr3},
cannot be summarized as an expression parameterized by an induction variable $k$.
\mdelete{In consequence, the repetitive states cannot be summarized and it will continue to generate such repetitive (non-equivalent/non-summarizable) state sequence.}\mrevise{According to MR2, to merge such repetitive states, we have to rely on widening operators, which are sound but imprecise~\cite{cousot1977abstract}.}
Recall that, 
in Figure~\ref{fig:motivating-example-2}(f) where SR3 is used,
the value of \textit{tok} is a sequence of $\tau^1$, $\tau^2$, $\tau^3$, and so on. Thus,
we can precisely summarize its value as $\tau^{k+1}$ \mrevise{via MR2}.

\section{Formalizing the Approach}\label{sec:approach}

\mreplace{This section formalizes our idea.
The formula}{In this section, the notation}
$a[b/c]$ returns the expression $a$ after using $b$ to replace all occurrences of $c$ in $a$.
We use $\textup{sat}(\phi)$ and $\textup{unsat}(\phi)$ to
mean that the constraint $\phi$ is satisfiable and not.
An ite($v_1, v_2, v_3$) formula returns \mdelete{$v2$}\mrevise{$v_2$} and $v_3$ if the condition $v_1$ is true and false, respectively.
We use a simplification procedure~\cite{dillig2010small},
$\phi_1' = \textup{simplify}(\phi_1, \phi_2)$,
to simplify $\phi_1$ but keep the equivalence of $\phi_1$ and $\phi_1'$
in terms of $\phi_2 \Rightarrow(\phi_1\equiv \phi_1')$.
\delete{We use $\phi_1', \phi_2' = \textup{cross-simplify}(\phi_1, \phi_2)$ to mean $\phi_1' = \textup{simplify}(\phi_1, \phi_2)$ and $\phi_2' = \textup{simplify}(\phi_2, \phi_1)$.}\remark{directly use simplify() without creating a new notation cross-simplify()}

\defparbf{Abstract Language.}
For clarity, we use a C-like language in Figure~\ref{fig:design_lang} to model 
a parser that implements an FSM via a do-while loop.
We use a do-while loop as it is a general form of loops with initialization, i.e., $\mathcal{S}; \textbf{while}(1) \{\mathcal{S};\}$.
The statements
could be assignments, binary operations,
read statements that read the next byte of a message to parse,
exit statements that exit the loop, and
branching statements that are uniquely identified by the identifier $\kappa$.
To use our approach,
users manually annotate the statement reading the inputs, e.g., the read function. 
The rest is fully automated.
Although we do not include function calls or returns for simplicity, 
our system is interprocedural as a call statement is equivalent to assignments from the actual parameters
to the formals, and a return statement is an assignment from the return value to its receiver.
The language abstracts away
pointer operations because the pointer analysis is not our
technical contribution and, in the implementation, we follow
existing works to resolve pointer relations~\cite{xie2005scalable}.
We do not assume nested loops for simplicity as we focus on the outermost loop that implements the FSM.
In practice, we observe that inner loops often serve for parsing repetitive fields in a network message rather than implementing the FSM. Hence, in the implementation, we follow traditional techniques to analyze inner loops~\cite{godefroid2011automatic,saxena2009loop}.

\begin{figure}[t]
    \centering
    \footnotesize
    \begin{tabular}{rcll}
        \textit{Parser}~$\mathcal{P}$ &:=&$\textbf{do}~\{~\mathcal{S};~\}~\textbf{while}(1);$ &~\\
        \textit{Statement}~$\mathcal{S}$ &:=& $v_1 \leftarrow v_2$ & \textbf{~::assign}\\
        & ~ & $~|~v_1 \leftarrow v_2 \oplus v_3$ & \textbf{~::binary}\\
        & ~ & $~|~v_1 \leftarrow \textit{read}()$& \textbf{~::read}\\
        & ~ & $~|~\textit{exit}()$& \textbf{~::exit}\\
        & ~ & $~|~\textbf{if}_\kappa~(v)~\{ \mathcal{S}_1; \}~\textbf{else}~\{ \mathcal{S}_2; \}$& \textbf{~::branching}\\
        & ~ & $~|~\mathcal{S}_1;\mathcal{S}_2$ & \textbf{~::sequencing}\\
        &&&\\
        & \multicolumn{3}{l}{$\oplus\in \{ \land, \lor, +, -, >, <, =, \ne, \dots \}$}
    \end{tabular}
    \vspace{-2mm}
    \caption{Language of target programs.}
    \label{fig:design_lang}
\end{figure}

\begin{figure}\footnotesize\centering
    \begin{tabular}{rcll}
        \textit{Abstract Value}~$\tilde{v}$ &:=& $c$ & \textbf{~::constant value}\\
        & ~ & $~|~\sigma^k$ & \textbf{~::current input of length $k$}\\
        & ~ & $~|~\tau^k$ & \textbf{~::previous input of length $k$}\\
        & ~ & $~|~\tilde{v}_1 \oplus \tilde{v}_2$& \textbf{~::binary formula}\\
        & ~ & $~|~\textup{ite}(\tilde{v}_1, \tilde{v}_2, \tilde{v}_3)$& \textbf{~::if-then-else formula}\\
        & ~ & $~|~\textup{int}(c_1, c_2)$& \textbf{~::interval}
    \end{tabular}
    \vspace{-2mm}
    \caption{Abstract values.}
    \vspace{-4mm}
    \label{fig:design_absval}
\end{figure}

\begin{figure*}[t]
    \centering
    \begin{minipage}[t]{\textwidth}\footnotesize
        $$
        \inference{
            \mathbb{I}\textup{ uses }\tau^k_i, \sigma^{k'}_{i'}
        }
        {
            \mathbb{I}, \phi
            \vdash:
            \mathbb{I}[\tau^{k+k'}_{i}/\tau^{k}_{i}][\tau^{k+k'}_{k+i'}/\sigma^{k'}_{i'}], \textup{true}
        }[\textbf{Init}]
        \quad
        \inference{
            \mathbb{I}(v_2) = \tilde{v}_2
        }
        {
            \mathbb{I}, \phi \vdash v_1 \leftarrow v_2: \mathbb{I}\cup (v_1, \tilde{v}_2), \phi
        }[\textbf{Assign}]
        \quad
        \inference{
            \mathbb{I}(v_2) = \tilde{v}_2\enskip\enskip\enskip \mathbb{I}(v_3) = \tilde{v}_3
        }
        {
            \mathbb{I},  \phi \vdash v_1 \leftarrow v_2\oplus v_3: \mathbb{I}\cup\{ (v_1, \tilde{v}_2\oplus \tilde{v}_3) \}, \phi
        }[\textbf{Binary}]
        $$
        \\
        $$
        \inference{
            \mathbb{I}\textup{ uses }\sigma^{k}_{i}
        }
        {
            \mathbb{I}, \phi  \vdash v_1 \leftarrow \textit{read}(): \mathbb{I}[\sigma^{k+1}_i/\sigma^{k}_i]\cup \{(v_1, \sigma^{k+1}_k)\}, \phi
        }[\textbf{Read}]
       \quad
    \inference{
        ~
    }
    {
        \mathbb{I}, \phi \vdash \textit{exit}(): \mathbb{I}, \phi
    }[\textbf{Exit}]
        \quad
        \inference{
            \mathbb{I}_1, \phi_1\vdash \mathcal{S}_1: \mathbb{I}_2, \phi_2
            \enskip\enskip\enskip
            \mathbb{I}_2, \phi_2 \vdash \mathcal{S}_2: \mathbb{I}_3, \phi_3
        }
        {
             \mathbb{I}_1, \phi_1 \vdash \mathcal{S}_1;\mathcal{S}_2: \mathbb{I}_3, \phi_3
        }[\textbf{Sequencing}]
        $$
        \\
        $$
        \inference{
            \mathbb{I}(v) = \tilde{v}
            \enskip
            \enskip
            \enskip 
            \mathbb{I}\cup\{(\kappa, \tilde{v})\}, \phi\land \kappa
            \vdash \mathcal{S}_1: 
            \mathbb{I}_1,  \phi \land \phi_1
            \enskip
            \enskip
            \enskip
            \mathbb{I}\cup\{(\kappa, \tilde{v})\}, \phi\land \lnot \kappa
            \vdash \mathcal{S}_2: 
            \mathbb{I}_2, \phi \land \phi_2
        }
        {
            \mathbb{I}, \phi
            \vdash 
            \textbf{if}_\kappa~(v)~\{ \mathcal{S}_1; \}~\textbf{else}~\{ \mathcal{S}_2; \}: 
            \left\{
            \begin{array}{lr}
                \mathbb{I}_1, \phi\land\phi_1 &  \textup{unsat}(\phi[\mathbb{I}(\kappa')/\kappa'] \land \lnot \tilde{v})\\
                \mathbb{I}_2, \phi\land\phi_2 & \textup{unsat}(\phi[\mathbb{I}(\kappa')/\kappa'] \land \hspace{1.6mm} \tilde{v})\\
                \{ (u, \textup{ite}(\tilde{v}, \tilde{u}_1, \tilde{u}_2) : (u, \tilde{u}_1)\in \mathbb{I}_1 \land (u, \tilde{u}_2)\in \mathbb{I}_2) \}, \phi\land(\phi_1\lor\phi_2) &  \textit{otherwise}
            \end{array}
            \right.
        }[\textbf{Branching}]
        $$
    \end{minipage}
    \caption{Transfer functions as inference rules for analyzing a loop iteration.}
    \vspace{-2.5mm}
    \label{fig:semantics_basic}
\end{figure*}

\defparbf{Abstract Domain.}
An abstract value of a variable represents all possible concrete values that may be assigned to the variable during program execution.
The abstract domain specifies the limited forms of an abstract value.
In our analysis,
the abstract value of a variable $v$ is denoted as $\tilde{v}$ and defined in Figure~\ref{fig:design_absval}.
An abstract value could be a constant value $c$ and a byte stream of length $k$, i.e., $\sigma^k$ and $\tau^k$,
which respectively represent the input byte stream read in the current loop iteration
and the previous iterations.\delete{We use $\sigma^n_i$ and $\sigma^n_{i..j}$ to represent the $(i+1)$th byte and a subsequence of $\sigma^n$, respectively.
We simply write $\sigma$ as a shorthand of $\sigma^1_0$.}\remark{Since we have defined this before, we do not repeat it here according to reviewers' comments.}
The symbols $\tau^n_i$, $\tau^n_{i..j}$, and $\tau$ are defined similarly \revise{as $\sigma^n_i$, $\sigma^n_{i..j}$, and $\sigma$}.
An abstract value can also be a first-order logic formula over other abstract values.
To ease the explanation, we only support binary and \textup{ite} formulas.
\revise{Especially, we also include an interval abstract value to mean a value between two constants. As discussed later in Algorithm~\ref{alg:mr2}, such interval abstract values allow our analysis to fall back to conventional interval-domain abstract interpretation~\cite{cousot1977abstract}, in order to guarantee convergence and~soundness.}

\defparbf{Abstract Interpretation.}
The abstract interpretation is described as transfer functions of \mrevise{each} program statement\mdelete{s in Figure~\ref{fig:design_lang}}. Each transfer function updates the program environment $\mathbb{E}= (\mathbb{I}, \phi)$.
Given the set $\bV$ of program variables and the set $\tilde{\bV}$ of abstract values,
$\mathbb{I}: \bV\mapsto \tilde{\bV}$ maps a variable to its abstract value.
\delete{The sets $\bM\subseteq \bV$ and $\bR\subseteq \bV$ implement a Mod/Ref analysis.
When analyzing a loop iteration,
the set $\bM$ records the variables revised in the current iteration,
and the set $\bR$ records a variable if it is used in the current iteration but references a value assigned to the variable in previous iterations.}\remark{The Mod/Ref analysis is not used and, thus, is removed.}
The constraint $\phi$ captures the \mrevise{skeletal} path constraint\mrevise{, which stands for a path set executed in a single loop iteration}\mdelete{of a loop iteration}.
We say $\phi$ is a skeletal path constraint
    because it is 
    in a form of conjunction or disjunction
    over the symbols $\kappa$ or $\lnot \kappa$, e.g., $\kappa_1 \land (\kappa_2 \lor \lnot \kappa_2)$,
    where each symbol $\kappa$ uniquely identifies a branch
    and is not evaluated to its branching condition.
    The real path constraint is denoted by the uppercase Greek letter $\Phi = \phi[\mathbb{I}(\kappa)/\kappa]$ where each $\kappa$ is replaced by its abstract value.
\mrevise{We list the transfer functions in Figure~\ref{fig:semantics_basic},
    which describe how we analyze a loop iteration, i.e., the procedure \texttt{abstract\_interpretation} in Algorithm~\ref{alg:framework}.
    In these transfer functions,
    we use $\mathbb{E} \vdash \mathcal{S}: \mathbb{E}'$ to describe the environment before and after a statement.}\mdelete{Figure~\ref{fig:semantics_basic} lists the transfer functions as inference rules.
    The rules describe how we analyze a loop iteration, i.e., \texttt{abstract\_interpretation}
    in Algorithm~\ref{alg:framework}.}

\mreplace{For initialization, we assume}{To initialize the analysis of a loop iteration, we set} the initial environment \mreplace{is}{to} $\mathbb{E} = (\mathbb{I}, \phi)$, which is obtained from the previous iteration\mrevise{,} and \mrevise{assume that} abstract values in $\bI$ use the symbols $\tau^k_i$ and $\sigma^{k'}_{i'}$.
This means that the previous iteration depends on an input stream of length $k+k'$, in which
$k$ bytes from iterations before the last iteration and $k'$ bytes from the last iteration.
For the current iteration,
all $k+k'$ bytes are from previous iterations.
Hence, we rewrite all $\sigma$ to $\tau$\mdelete{ for initialization}.

The rules for assignment, binary operation, read, and exit are straightforward,
which update the abstract value of a variable.
\delete{During the analysis,
if a variable is modified, it is added to the set $\bM$;
and if a referenced variable is not in the set $\bM$, its value may come from previous iterations
and the variable is added to the set $\bR$.}The sequencing rule says that, for two consecutive statements,
we analyze them in order.
The branching rule\delete{, Br-I, Br-II, and Br-III,} states how we handle conditional statements.\delete{The rules Br-I and Br-II
deal with the case where only one branch is feasible.
The rule Br-III deals with the case where both branches are feasible.}\remark{we merge the original Br-I, Br-II, and Br-III rules into a single branching rule.}
In the branching rule, $(\mathbb{I}, \phi)$ represents the environment before a branching statement.
$(\mathbb{I}_1, \phi\land\phi_1)$ and $(\mathbb{I}_2, \phi\land\phi_2)$
are program environments we respectively infer from the two branches.
At the joining point,
we \revise{either use the analysis results of one branch if the other branch is infeasible, or} merge program environments from both branches.
\revise{When merging results from both branches, v}\delete{V}ariables assigned different values from the two branches are merged via the \textit{ite} operator.
\delete{The sets $\bM_1$ and $\bM_2$ are merged via set intersection as the value of a variable must be from the current iteration only when both branches modify the variable.
The sets $\bR_1$ and $\bR_2$ are merged via set union as a variable references a value from previous iterations if its value in one branch comes from previous iterations.}Path constraints are merged via disjunction with the common prefix pulled out.


\remark{To support formal discussion of soundness, we rewrite the remaining part of this section. The basic idea is~not changed. Old text can be found in the draft submitted before.}

\defparbf{Abstract Finite State Machine.}
\mrevise{We use a graph structure~to represent an FSM.
That is, an FSM is a set of labeled edges. Each edge is a triple $(S, \bE_S, S')$ where  = $\bE_S = (\bI_S, \phi_S)$,
meaning a transition from the state $S$ to the state $S'$ with the transition constraint $\phi_S[\mathbb{I}_S(\kappa)/\kappa]$. In the triple, $\bE_S$ is the resulting program environment after analyzing the path set $S$ in a loop iteration. Next, 
we formally describe the other two key procedures, i.e., \texttt{split} and \texttt{merge}, in Algorithm~\ref{alg:framework}.}

\defparit{(1) Splitting Rules (SR1-3).}
Splitting a state consists of two steps ---
splitting the path set the state represents and recomputing its outgoing program environment.

SR1 splits two overlapping path sets $S_1$ and $S_2$ into at~most three subsets,
respectively represented by
$\phi_{S_1}\land\lnot \phi_{S_2}$ that means paths in the first set but not in the second,
$\phi_{S_1}\land \phi_{S_2}$ that means paths shared by the two sets,
and
$\lnot \phi_{S_1}\land \phi_{S_2}$ that means paths not in the first set but in the second.
We create a state for each of the three skeletal constraints if it is satisfiable.
SR2 and SR3 isolate some special paths from a path set.
Given the path set $S_1$ and the paths $S_2$ to isolate,
we create two states represented by $\phi_{S_1}\land\lnot \phi_{S_2}$ and $\phi_S\land\phi_{S_2}$, respectively. 

\begin{algorithm}[t]\footnotesize
    \caption{Splitting Rules (SR1-3).}
    \label{alg:splitting}
    \SetKwFunction{Split}{split}
    \SetKwProg{Proc}{Procedure}{}{}
    \Proc{\Split{$(S_1, \bE_{S_1}, S_2), (S_2, \bE_{S_2}, S_3)$}}{
        \textbf{assume} $\bE_{S_1} = (\bI_{S_1}, \phi_{S_1})$ and $\bE_{S_2} = (\bI_{S_2}, \phi_{S_2})$\;
        \textbf{assume} $S_2$ is split into two sub-states $S_{21}$, $S_{22}$, $\dots$\;
        \textbf{let} $\Phi_{S_{2i}} = \phi_{S_{2i}}[\mathbb{I}_{S_2}(\kappa)/\kappa]$\;
        \textbf{let} $\bI_{S_{2i}} = \bI_{S_2}[\textup{simplify}(\tilde{v}, \Phi_{S_{2i}})/\tilde{v}]$\; 
        \textbf{let} $\bE_{S_{2i}} = (\bI_{S_{2i}}, \phi_{S_{2i}})$\;
        {\fontdimen2\font=1pt replace input transitions with} $(S_1, \bE_{S_{1}}, S_{2i}), (S_{2i}, \bE_{S_{2i}}, S_3)$\;
    }
\end{algorithm}

After a state is split into multiple sub-states,
we recompute the outgoing program environment for each sub-state. 
\mrevise{Algorithm~\ref{alg:splitting} and Figure~\ref{fig:splitting} show the splitting procedure,
where we assume we split the state $S_2$ into multiple sub-states $S_{2i}$
and split its outgoing transition $(S_2, \bE_{S_2}, S_3)$ into $(S_{2i}, \bE_{S_{2i}}, S_3)$.
The splitting procedure consists of two steps.
First, Line~4 in Algorithm~\ref{alg:splitting} computes the real path constraint according to the skeletal path constraint of each sub-state.
Second, Line~5 recomputes each abstract value under the new path constraint.
Basically, this step is to remove values from unreachable branches.
For instance, assume $\bI_{S_2}(v) = \textup{ite}(\tilde{v}_1, \tilde{v}_2, \tilde{v}_3)$,
meaning that after analyzing the path set $S_2$, the abstract value of the variable $v$
is either $\tilde{v}_2$ or $\tilde{v}_3$, depending on if the branching condition $\tilde{v}_1$ is true.
If paths in the subset $S_{21}$ ensures $\tilde{v}_1 = \textup{true}$,
we then rewrite the abstract value as $\bI_{S_{21}}(v) = \tilde{v}_2$.}

\defparit{(2) Merging Rules (MR1).}
MR1 merges two equivalent states.
\mstartreviseblock
Lines~13-14 of Algorithm~\ref{alg:framework} implements this rule.
We show the idea in Figure~\ref{fig:mr1},
where we assume $S_1'\equiv S_1$ and $\bE_{S_1'}\equiv\bE_{S_1}$.
In this case, we merge $S_1$ and $S_1'$,
but do not compute the next states using $\bE_{S_1'}$
because we have already computed them using its equivalent counterpart $\bE_{S_1}$.
Thus, Algorithm~\ref{alg:framework} does not add $(S_1', \bE_{S_1'})$ to the worklist
at Lines~13-14.
\mdismissreviseblock

\begin{figure}[t]
    \centering
    \includegraphics[width=\linewidth]{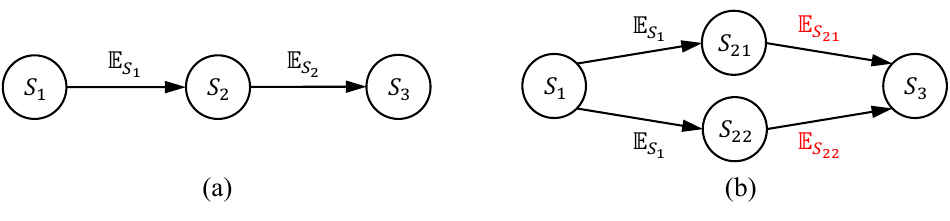}
    \vspace{-8mm}
    \caption{\revise{SR1-3. (a) Before splitting. (b) After splitting.}}\label{fig:splitting}
    \vspace{-4mm}
\end{figure}

\begin{figure}[t]
    \centering
    \includegraphics[width=\linewidth]{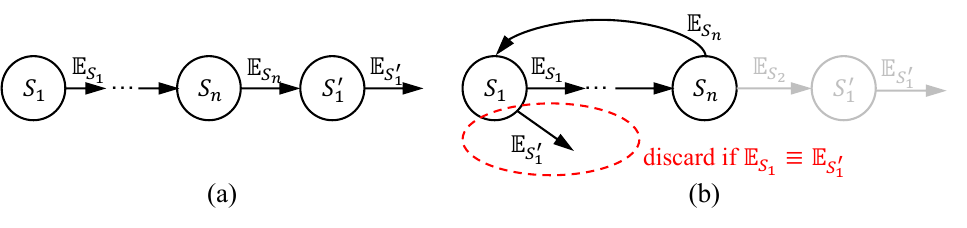}
    \vspace{-8mm}
    \caption{\revise{MR1. (a) Before merging. (b) After merging.}}\label{fig:mr1}
    \vspace{-2mm}
\end{figure}

\begin{figure}[t]
    \centering
    \includegraphics[width=\linewidth]{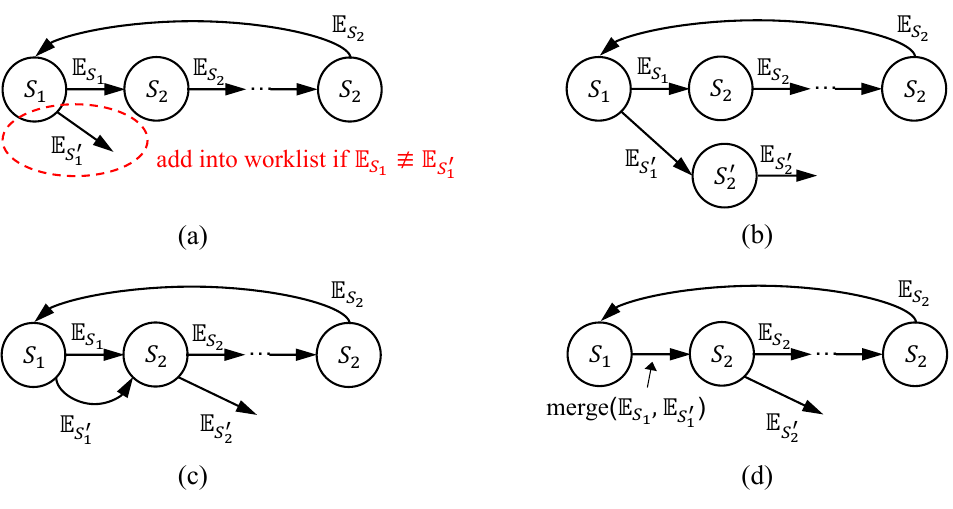}
    \vspace{-8mm}
    \caption{\revise{MR2. (a) Add $(S_1, \bE_{S_1'})$ into worklist if $\bE_{S_1'}\not\equiv\bE_{S_1}$. (b) Generate $(S_2', \bE_{S_2'})$ using $\bE_{S_1'}$ as the precondition. (c) Merge $S_2$ and $S_2'$. (d) Merge transitions between $S_1$ and $S_2$.}}\label{fig:mr2}
    \vspace{-4mm}
\end{figure}

\begin{algorithm}[t]\footnotesize
    \caption{Merging Rules (MR2).}
    \label{alg:mr2}
    \SetKwFunction{Merge}{merge}
    \SetKwFunction{Interval}{interval}
    \SetKwFunction{Widen}{widen}
    \SetKwProg{Proc}{Procedure}{}{}
    \Proc{~\Merge{$\bE_1$, $\bE_2$}}{
        \textbf{assume} $\bE_{1} = (\bI_{1}, \phi)$ and $\bE_{2} = (\bI_{2}, \phi)$\;
        \textbf{let} $\Phi_{S_1} = \phi_{S_1}[\mathbb{I}_{S_1}(\kappa)/\kappa]$;
        \enskip $\bI_{1} = \bI_{1}[\Interval(\tilde{v}, \Phi_{S_1})/\tilde{v}]$\;
        \textbf{let} $\Phi_{S_2} = \phi_{S_2}[\mathbb{I}_{S_2}(\kappa)/\kappa]$;
        \enskip $\bI_{2} = \bI_{2}[\Interval(\tilde{v}, \Phi_{S_2})/\tilde{v}]$\;
        \ForEach{$v$ \textup{such that} $\bI_{1}(v) = \tilde{v}_1~\land~\bI_{2}(v) = \tilde{v}_2$}{
            \textbf{let} $\bI(v) = \textup{widen}(\tilde{v}_1, \tilde{v}_2)$\;
        }
        \Return $(\bI, \phi)$\;
    }
    \Proc{~\Interval{$\tilde{v}$, $\Phi$}}{
        \textbf{let} $c_{min} =$ minimize $\tilde{v}$ with respect to $\Phi$ by SMT solver\;
        \textbf{let} $c_{max} =$ maximize $\tilde{v}$ with respect to $\Phi$ by SMT solver\;
        \Return int($c_{min}$, $c_{max}$)\;
    }
    \Proc{~\Widen{\textup{int(}$a_1$, $b_1$\textup{)}, \textup{int(}$a_2$, $b_2$\textup{)}}}{
        \textbf{let} $c_1 =$ ite($a_1>a_2$, $-\infty$, $a_1$); \enskip
        \textbf{let} $c_2 =$ ite($b_1<b_2$, $+\infty$, $b_1$)\;
        \Return int($c_{1}$, $c_{2}$)\;
    }
\end{algorithm}

\defparit{(3) Merging Rules (MR2).}
MR2 merges two states that represent the same path sets but have non-equivalent outgoing program environments.
\startreviseblock
Let us consider the example in Figure~\ref{fig:mr2} to understand how Algorithm~\ref{alg:framework} deals with this case.
Figure~\ref{fig:mr2}(a) is the same as Figure~\ref{fig:mr1}(b) except that we assume $\bE_{S_1'}\not\equiv\bE_{S_1}$.
In this situation, we add $(S_1, \bE_{S_1'})$ to the worklist (see Lines~13-14 in Algorithm~\ref{alg:framework}).
When $(S_1, \bE_{S_1'})$ is popped out, we will perform abstract interpretation using $\bE_{S_1'}$ as the initial program environment (see Lines~5-6 in Algorithm~\ref{alg:framework}).
Assume the abstract interpretation produces $(S_2', \bE_{S_2'})$ where $S_2'\equiv S_2$ as illustrated in Figure~\ref{fig:mr2}(b).
In Figure~\ref{fig:mr2}(c), we merge $S_2$ and $S_2'$, yielding multiple non-equivalent transitions between $S_1$ and $S_2$.
Lines~16-19 in Algorithm~\ref{alg:framework} merge such transitions, yielding Figure~\ref{fig:mr2}(d).
If the merged environment, i.e., $\textup{merge}(\bE_{S_1}, \bE_{S_1'})$ equals $\bE_{S_1}$ or $\bE_{S_1'}$,
we do not add $(S_1, \textup{merge}(\bE_{S_1}, \bE_{S_1'}))$ to the worklist because the resulting transition $(S_1, \bE_{S_1}, S_2)$ or $(S_1, \bE_{S_1'}, S_2)$ has been in the FSM.
Otherwise, the pair $(S_1, \textup{merge}(\bE_{S_1}, \bE_{S_1'}))$ will be added to the worklist for further computation.

A na\"ive merging procedure is shown in Algorithm~\ref{alg:mr2},
which utilizes the traditional interval abstract domain to guarantee soundness and convergence.
Lines~3-4 convert each abstract value to an interval, int($c_{min}$, $c_{max}$), by
solving two optimization problems via an SMT solver.
Basically, solving the optimization problems respectively produces the minimum and maximum solutions, $c_{min}$ and $c_{max}$, of the abstract value $\tilde{v}$ with respect to the path constraint.
Lines~5-6 merge the interval values via the traditional widening operator~\cite{cousot1977abstract}.
As proved by Cousot and Cousot~\cite{cousot1977abstract},
the widening operator ensures convergence and soundness,
which, in our context, means that it ensures the convergence and soundness of computing a fixed-point transition between two states.
Nonetheless, the na\"ive merging procedure could result in a significant loss of precision because both the computation of intervals (Lines 3-4) and the merging of intervals (Lines 5-6) over-approximate each abstract value.
Thus, before using the interval abstract domain to merge transitions,
we always try an induction-based solution, which is discussed below. 

The induction-based solution is sound and does not lose precision~\cite{angluin1983inductive}. In the solution,
we delay the transition merging operation until the number of transitions between a pair of states reaches a predefined constant.
For instance, 
in Figure~\ref{fig:mr2-induction}(a),
we do not merge transitions until the number of transitions between each pair reaches 3.
Given a list of transitions between a pair of states,
we can then perform the inductive inference in two steps -- guess and check.
For instance, 
in Figure~\ref{fig:mr2-induction}(a),
assume $\bI_{11}(v) = \sigma + 1$,  $\bI_{12}(v) = \sigma + 2$ and $\bI_{13}(v) = \sigma + 3$.
As shown in Figure~\ref{fig:mr2-induction}(b),
we then inductively ``guess'' the $k$th abstract value of the variable $v$ as $\bI_{1k}(v) = \sigma + k$.
To check the correctness of $\bI_{1k}(v) = \sigma + k$,
as shown in Figure~\ref{fig:mr2-induction}(c),
we rerun the abstract interpretation using $\bE_{S_{nk}}$ as the initial program environment,
if in the resulting program environment, the abstract value of $v$ is $\sigma + (k+1)$,
it means the summarized value $\bI_{1k}(v) = \sigma + k$ is correct.
This guess-and-check procedure follows the procedure of mathematical induction~\cite{schmidt1973introduction}
and, thus, is correct.
\dismissreviseblock

\begin{figure}[t]
    \centering
    \vspace{-2mm}
    \includegraphics[width=\linewidth]{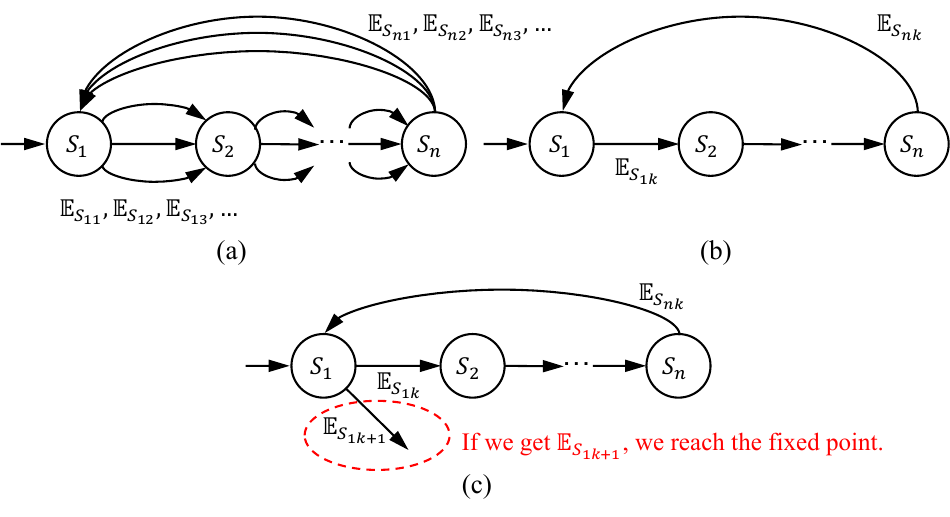}
    \vspace{-8mm}
    \caption{\revise{MR2 via induction. $\bE_{S_{ij}} = (\bI_{S_{ij}}, \phi_{S_i})$. (a) Delay merging. (b) Guess. (c) Fixed-point computation.}}\label{fig:mr2-induction}
    \vspace{-4mm}
\end{figure}

\begin{algorithm}[t]\footnotesize
    \caption{Merging Rules (MR3).}
    \label{alg:mr3}
    \SetKwFunction{Merge}{merge}
    \SetKwProg{Proc}{Procedure}{}{}
    \Proc{\Merge{$(S_1, \bE_{S_1}, S_2), (S_2, \bE_{S_2}, S_3)$}}{
        \textbf{assume} $\bE_{S_1} = (\bI_{S_1}, \phi_{S_1})$ and $\bE_{S_2} = (\bI_{S_2}, \phi_{S_2})$\;
        \textbf{let} $\Phi_{S_1} = \phi_{S_1}[\mathbb{I}_{S_1}(\kappa)/\kappa]$;
        \enskip $\Phi_{S_2} = \phi_{S_{2}}[\mathbb{I}_{S_2}(\kappa)/\kappa]$\;
        \textbf{let} $\Phi_{S_1} = \textup{simplify}(\Phi_{S_1}, \Phi_{S_{2}})$;
        \enskip $\Phi_{S_{2}} = \textup{simplify}(\Phi_{S_{2}}, \Phi_{S_1})$\;
        \textbf{if }{$\Phi_{S_{2}}$ \textup{does not use any symbol} $\tau$~} \textbf{then} {\Return\;}
        \textbf{assume} $\Phi_{S_{1}} = f(\sigma^k)$\;
        \If{$\Phi_{S_{2}} = g(\sigma^l) \land h(\tau^m)$} {
            \textbf{let} $\Phi_{S_1} = f(\sigma^k) \land h(\tau^m)[\sigma^k_{i-m+k}/\tau^m_{i\ge m-k}][\tau^{m-k}_i/\tau^m_{i<m-k}]$\;
            \textbf{let} $\Phi_{S_2} = g(\sigma^l)$\;
        }
        \ElseIf{$\Phi_{S_{2}} = g(\sigma^l) \lor h(\tau^m)$} {
            split the state $S_2$ as shown in Figure~\ref{fig:mr3}(c-d) and
            recursively call this procedure.
        }
        \Else{
            \textbf{let} $\Phi_{S_1} = f(\sigma^k)[\sigma^{k+l}_i/\sigma^k_i];\enskip \Phi_{S_2} = g(\sigma^l, \tau^m)[\sigma^{k+l}_{k+i}/\sigma^l_i]$\;
            \If{$m\ge k$} {
            \textbf{let} $\Phi = \Phi_{S_1} \land \Phi_{S_{2}}[\sigma^{k+l}_{i-m+k}/\tau^m_{i\ge m-k}][\tau^{m-k}_i/\tau^m_{i<m-k}]$\;
            }
            \textbf{else\enskip let} $\Phi = \Phi_{S_1} \land \Phi_{S_{2}}[\sigma^{k+l}_{k-m+i}/\tau^m_{i}]$\;
            merge transitions into one from $S_1$ to $S_3$ constrained by $\Phi$\;
        }
    }
\end{algorithm}

\begin{figure}[t]
    \centering
    \includegraphics[width=\linewidth]{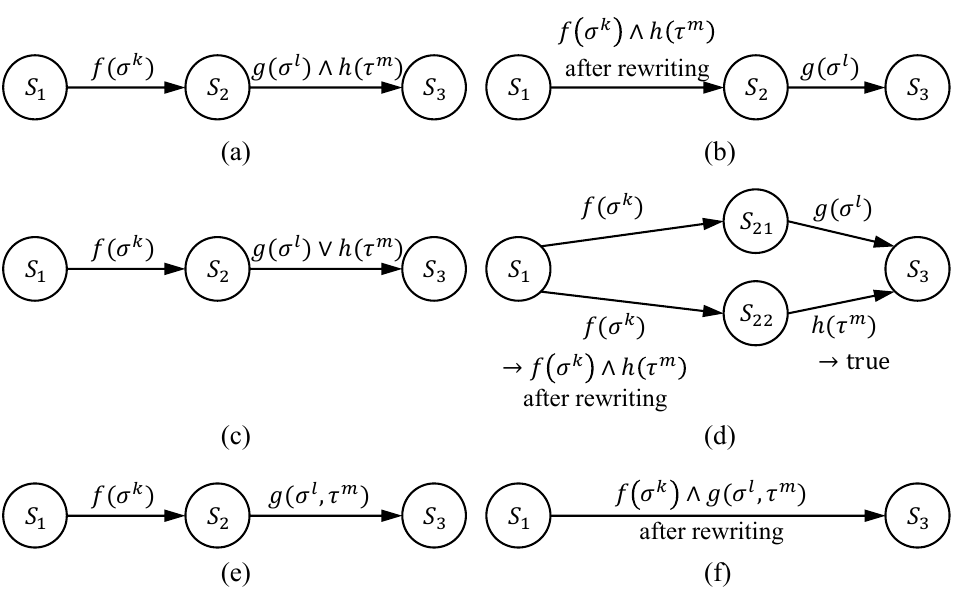}
    \vspace{-6mm}
    \caption{\revise{MR3. Eliminating $\tau$ in (a-b) conjunctive constraints, (c-d) disjunctive constraints, and (e-f) constraints where $\tau$ cannot be isolated by disjunction or conjunction.}}\label{fig:mr3}
    \vspace{-2mm}
\end{figure}

\defparit{(4) Merging Rules (MR3).}
\mrevise{MR3 ensures the validity of FSM by eliminating state transitions that refer to inputs consumed by previous transitions. It is performed after an FSM is produced by Algorithm~\ref{alg:framework}.
Algorithm~\ref{alg:mr3} and Figure~\ref{fig:mr3}
demonstrate how it works on two transitions,
one is from the state $S_1$ to the state $S_2$ and consumes $k$ bytes, i.e., $\sigma^k$;
the other is from the state $S_2$ to the state $S_3$, consumes $l$ bytes, i.e., $\sigma^{l}$, and, meanwhile, constrains $m$ bytes consumed by previous transitions, i.e., $\tau^{m}$.
First, for conjunctive constraints, e.g., $g(\sigma^l)\land h(\tau^m)$ in Figure~\ref{fig:mr3}(a),
we only need to move the constraint $h(\tau^m)$ to the previous transition
and perform constraint rewriting.
Such rewriting does not change the semantics of the transition constraint but just lets it follow the definitions of $\sigma$ and $\tau$.
Second, for disjunctive constraints, e.g., $g(\sigma^l)\lor h(\tau^m)$ in Figure~\ref{fig:mr3}(c),
we split the state $S_2$ to eliminate the disjunctive operator as shown in Figure~\ref{fig:mr3}(d)
and then use the method for conjunction discussed above.
Third, for constraints that cannot isolate $\tau$-related sub-formulas via disjunction or conjunction,
as shown in Figure~\ref{fig:mr3}(f), we merge the transitions into one.}

\startreviseblock
\begin{theorem}[Soundness and Completeness]\label{th:sound}
    Given a program in the language defined in Figure~\ref{fig:design_lang},
    Algorithm~\ref{alg:framework} is sound using the aforestated splitting and merging rules.
        It is complete if the interval domain is never used during the analysis. 
\end{theorem}
\begin{proof}
   The proof is discussed in Appendix~\ref{app:soundness}.
\end{proof}
\dismissreviseblock

\ifdiff
\else
\newpage
\fi

\noindent
\textbf{Discussion}.
\revise{We propose a static analysis that can infer an FSM from a parsing loop.
    While it is undecidable to check if an input loop intends to implement an FSM,
    as discussed in Theorem~\ref{th:sound},
    given any loop in our abstract language,
    our approach guarantees to output a sound result.}
\revise{Nevertheless, the implementation in practice shares some common limitations with general static analysis.}
For instance,
our static analysis is currently implemented for C programs
and does not handle virtual tables in C++.
We focus on source code and do not handle inline assembly. For libraries without available source code, e.g., crc16() and md5(), which are widely used to compute checksums or encrypt messages, we manually model these APIs.
A common limitation shared with the state of the art is that,
if the code implements a wrong FSM, the FSM we infer will be incorrect, either.
Nevertheless,
we will show that our approach is promising via a set of experiments.

\section{Evaluation}\label{sec:evaluation}

On top of the LLVM compiler framework~\cite{lattner2004llvm} and the Z3~theorem prover~\cite{de2008z3},
we have implemented \tool\ for protocols written in C.
The source code of a protocol is compiled into the LLVM bitcode
and sent to \tool\ for inferring the FSM.
In \tool, 
LLVM provides facilities to manipulate the code and 
Z3 is used to represent abstract values as symbolic expressions
and solve path constraints.

\defparbf{Research Questions.}
\mdelete{In the evaluation, we focus on three research questions.}First,
\mrevise{we compare our approach to the state-of-the-art static analysis for FSM inference, i.e., \proteus~\cite{xie2016proteus,xie2019automatic}}\mdelete{we study the speed of inferring an FSM and the size of each inferred FSM by comparing our approach to the state of the arts}.
Second,
\mrevise{we compare \tool\ to dynamic techniques, including ReverX~\cite{antunes2011reverse}, AutoFormat~\cite{lin2008automatic}, and Tupni~\cite{cui2008tupni}.}\mdelete{for the precision and recall of our static analysis, we investigate how many false state transitions \tool\ reports and how many true state transitions \tool\ misses.}
Third, 
to show the security impacts,
we apply \tool\ to fuzzing and applications beyond protocols.

\defparbf{Benchmarks.}
Our approach is designed to work on the C code that implements the FSM parsing loop for regular protocols.
We do not find any existing test suite that contains such C code.
Thus, we build the test suite.
To this end,
we search the Github for regular protocols implemented in C language via the keywords, ``protocol parser'', ``command parser'', and ``message parser'',
until we found the ten in Table~\ref{tab:basics}.
These protocols include
text protocols such as ORP and binary protocols such as MAVLINK.
They are widely used in different domains in the era of the internet of things.
For example,
ORP allows a customer asset to interact with Octave edge devices.
MAVLINK is a lightweight messaging protocol for communicating with drones.
TINY specifies the data frames sent over serial interfaces such as UART and~telnet. 
SML defines the message formats for smart meters.
RDB is a protocol for communicating with Redis databases. 
MQTT is an OASIS standard messaging protocol for IoT devices.
MIDI is for musical devices and KISS is for amateur radio.  


\defparbf{Environment.}
All experiments are conducted on a Macbook Pro (16-inch, 2019)
equipped with an 8-core 16-thread Intel Core i9 CPU with 2.30GHz speed and 32GB of memory. 

\subsection{\replace{Speed and Size}{Against Static Inference Techniques}}\label{subsec:eval_static}

Our key contribution is \mreplace{to provide a practical}{a} static analysis \mreplace{to infer the parsing FSM}{that infers FSMs} 
\mreplace{by avoiding the path-explosion problem}{without suffering from path explosion}.
To show the impacts of our design,
we run both \tool\ and the state-of-the-art technique, \proteus~\cite{xie2016proteus,xie2019automatic}, against the benchmark programs on a 3-hour budget per program.
\revise{The time cost of each analysis is shown in Figure~\ref{fig:speed} in log scale.
As illustrated,
\proteus\ cannot complete many analyses within the time limit due to path explosion.
By contrast, all our analyses finish in five minutes,
exhibiting at least 70$\times$ speedup compared to \proteus.}
\revise{Since both \proteus\ and \tool\ perform path-sensitive analysis, 
they have the same precision and recall when both of them succeed in inferring the FSM for a protocol, e.g., ORP. We detail the results of precision and recall in \S\ref{subsec:eval_dynamic}}

\begin{table}[t]
    \footnotesize
    \centering
    \renewcommand{\arraystretch}{0.9}
    \caption{Sizes of the Inferred State Machines}~\\
    \label{tab:basics}
    \begin{tabular}{c|cc|cc}
        \toprule
        \multirow{2}{*}{\textbf{Protocols}} & \multicolumn{2}{c|}{\textbf{\tool}} & \multicolumn{2}{c}{\textbf{\proteus}} \\
         & \# states        & \# transitions        & \# states  & \# transitions \\
        \midrule
        ORP~\cite{orp}               &       5           &          8             &    42        &      92          \\
        MAVLINK~\cite{mavlink}           &       42           &          197             &      -      &     -           \\
        IHEX~\cite{ihex}              &             15     &           63            &         -   &        -        \\
        BITSTR~\cite{bitstring}              &     22             &           75            &       -     &      -          \\
        TINY~\cite{tiny}              &         14         &             54          &       151     &      872          \\
        SML~\cite{sml}               &        32          &           89            &     -       &         -       \\
        MIDI~\cite{midi}              &           19       &              81         &     765       &       3812         \\
        MQTT~\cite{mqtt}              &        28          &        87               &    105        &      581          \\
        RDB~\cite{rdb}               &         22         &         57              &        -    &            -    \\
        KISS~\cite{kiss}              &          6        &           12            &         24   &          142     \\
        \bottomrule
    \end{tabular}
\end{table}

\begin{figure}[t]
    \centering
    \includegraphics[width=0.76\columnwidth]{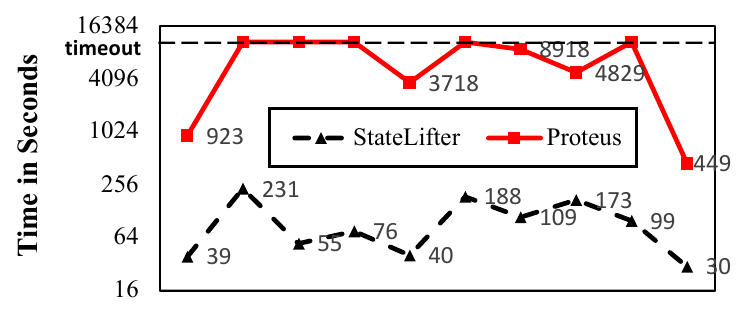}
    \vspace{-2mm}
    \caption{Time cost. The X-axis lists the ten protocols.}\label{fig:speed}
    \vspace{-2mm}
\end{figure}

Table~\ref{tab:basics} shows the size of each inferred FSM by \tool\ and \proteus.
Observe that the FSMs inferred by our approach are much ($4\times$-$40\times$) smaller than those inferred by \proteus.
It demonstrates that our design not only significantly mitigates the path explosion problem but also infers highly compressed FSMs, which can be expected to be easier to use in practice.

\subsection{\replace{Precision and Recall}{Against Dynamic Inference Techniques}}\label{subsec:eval_dynamic}

\replace{To evaluate the precision and recall of our approach,
    we manually compared every FSM we generate with the ground truth, i.e., the official specifications.
    The precision is measured by
    the ratio of the correctly inferred state transitions to all inferred state transitions in each FSM.
    The recall is measured by
    the ratio of the correctly inferred state transitions to the total number of transitions established according to the ground truth.}{Dynamic analysis is orthogonal to static analysis. Thus, in general, they are not comparable. Nevertheless, for the purpose of reference rather than comparison, we evaluate three dynamic analyses, including ReverX~\cite{antunes2011reverse}, AutoFormat~\cite{lin2008automatic}, and Tupni~\cite{cui2008tupni}.
    ReverX is a black-box approach that learns an FSM from input messages without analyzing the code.
    It instantiates general automata induction techniques like L*~\cite{angluin1987learning} and is specially designed for protocol format inference.
    Auto-Format and Tupni are white-box approaches that rely on dynamic dataflow analysis.
    They generate message formats in BNF, which can be easily converted to FSMs.
    Given that all analyses can complete within a few minutes, our focus is primarily on examining their precision and recall.
    In Appendix~\ref{app:compare}, we discuss the details of how we compute precision and recall.
    Intuitively,
the precision is the ratio of correct state transitions to all inferred transitions;
and the recall is the ratio of correct state transitions to all transitions in the ground truth.}

\mdelete{For reference,
we also run two network protocol reverse engineering techniques, 
AutoFormat and Tupni, against the benchmark protocols.
As discussed before,
mainstream approaches in the field, including AutoFormat and Tupni, are dynamic program analyses rather than static analyses and lack strategies to handle FSM parsing loops.
Nevertheless, this experiment is designed only for reference but not for a comparative purpose as
dynamic program analysis and static program analysis work under different assumptions.
For instance,
to make AutoFormat and Tupni work, following their original works,}\mrevise{To drive the dynamic analyses,}
we randomly generate one thousand valid messages as their inputs. 
By contrast, our static analysis does not need any inputs and, thus, provides a promising alternative to the state of the art\mdelete{s} especially when the input quality cannot be guaranteed.
The precision and recall of the \mreplace{three reverse-engineering tools}{inferred FSMs} are plotted in Figure~\ref{fig:precision_recall}.
It shows that we achieve over 90\% precision and recall
while the others \mdelete{are much lower.
They }often generate over 40\% false or miss 50\% true transitions\mrevise{.}
\mreplace{as}{This is because} they depend on a limited number of input messages and cannot handle FSM parsing loops well.
\tool\ also reports a few false transitions or misses some true ones
as it inherits some general limitations of static analysis (see \S\ref{sec:approach}).

\begin{figure}[t]
    \centering
    \includegraphics[height=1.1in]{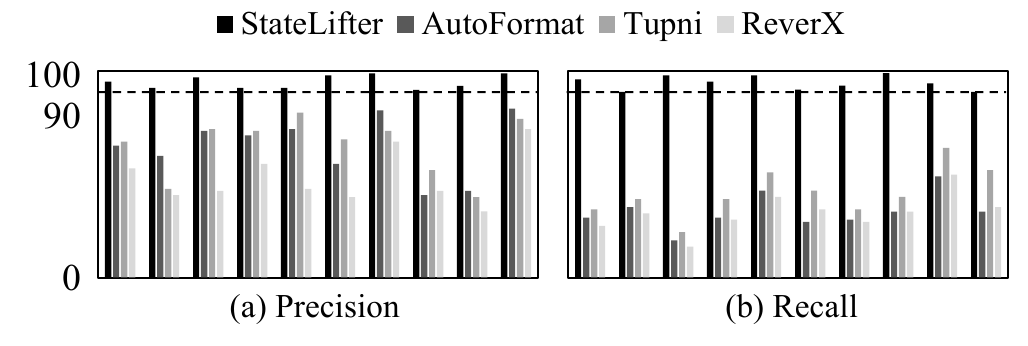}
    \vspace{-4mm}
    \caption{Precision and recall. \mrevise{X-axes list the ten protocols.} \remark{updated by adding ReverX}}\label{fig:precision_recall}
    \vspace{-4mm}
\end{figure}

\subsection{Security Applications}

\mdelete{To show the security impact,
we leverage the inferred FSMs to improve two typical protocol fuzzers, AFLNet and BooFuzz.
The former is a mutation-based fuzzer and the latter is a generation-based fuzzer.
We also discuss applications of our approach to other domains beyond network protocols.}\mremark{remove to save space}

\ifdiff
\smallskip
\else
\fi
\noindent\textbf{Protocol Fuzzing.}
AFLNet~\cite{pham2020aflnet} accepts a corpus of valid messages as the seeds 
and employs a lightweight mutation~method. 
Thus, 
we create a seed
corpus, where each message is generated by solving the transition constraints in the FSMs.
\mdelete{For }BooFuzz~\cite{boofuzz}\mdelete{, since it} directly accepts the message formats as its input and automatically generates messages\mreplace{ for fuzzing,}{. Thus,}
we respectively input the formats inferred by \tool, \revise{ReverX,} AutoFormat, and Tupni to BooFuzz.
The experiments are performed on a 3-hour budget and repeated 20 times to avoid random factors.
\mreplace{The results are}{As} shown in Figure~\ref{fig:fuzz}\mreplace{.
S}{, s}ince we can provide more precise and complete \mreplace{specifications}{formats},
fuzzers enhanced by \tool\ achieve \replace{$1.2\times$-$3.1\times$}{$1.2\times$-$3.3\times$} coverage.
\mreplace{The fuzzers enhanced by our approach}{Meanwhile, we} detect \mreplace{a dozen of}{twelve} zero-day bugs while the others
detected only two of them. \mrevise{We provide an example of detected bugs in Appendix~\ref{app:bugs}.}
All detected bugs are exploitable as they can be triggered via crafted messages.
\revise{Thus, they may pose a notable threat to software security in the industry. For example, we identified four vulnerabilities in the official implementation of ORP~\cite{orp}, which is commonly used for connecting Octave edge devices to the cloud~\cite{octave}.}

\begin{figure}[t]
    \centering
    \includegraphics[height=1.1in]{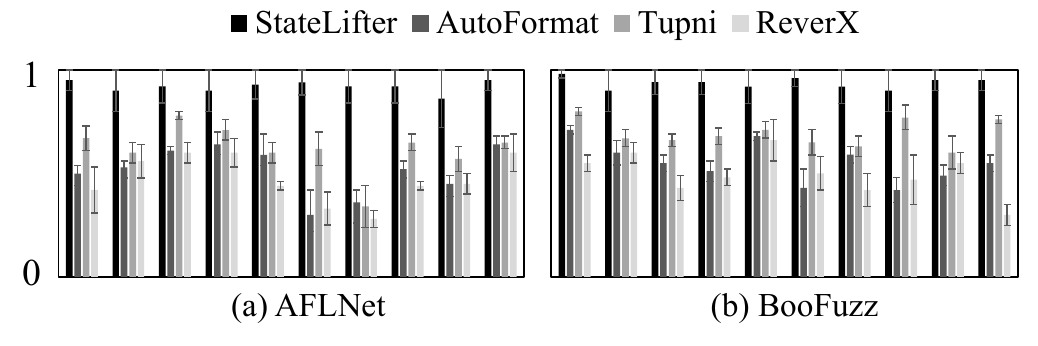}
    \vspace{-4mm}
    \caption{\mrevise{X-axes list the ten protocols.} Y-axes are coverage normalized to one with a 95\% confidence interval. \remark{updated by adding ReverX}}
    \label{fig:fuzz}
    \vspace{-1mm}
\end{figure}

\mdelete{Figure~\ref{fig:bug} shows a global-buffer-overflow bug detected in the protocol SML.
The parsing loop calls the function \textit{smlState} with an input byte
and updates the parsing state according to the byte.
In each parsing iteration, it pushes at most two bytes into the global buffer \textit{listBuffer}, of which the maximum length is 80.
This means that to trigger the bug in the function \textit{smlOBISManufacturer}, a message has to pass at least 40 iterations of the parsing loop.
In other words, the bug-triggering message must pass at least 40 state transitions in an inferred FSM.
This requires the inferred FSM to be of high precision and recall. Otherwise, ill-formed messages will be generated,
which are easy to be pruned and cannot execute deep program paths.
As discussed in \S\ref{sec:motivating_example}, conventional approaches cannot handle FSM parsing loops well. Hence, fuzzers armed with them miss this bug.}\mremark{move to the appendix to save space.}

%
\newpage
%

\section{Related Work}\label{sec:relatedwork}

\mdelete{In \S\ref{sec:problem_scope} and \S\ref{sec:motivating_example},
we have discussed how our approach is different from related work on dynamic-analysis-based reverse-engineering techniques
and static loop analyses. In what follows,
we discuss other two groups of related work.}\mremark{remove to save space}

\ifdiff
\smallskip
\else
\fi
\noindent
\textbf{Static Analysis for Protocol Reverse Engineering.}
While almost all existing works for inferring message formats use dynamic analysis,
Lim et al.~\cite{lim2006extracting} proposed a static analysis that is different from \tool\ in two aspects.
First, it infers the formats of output messages whereas we focus on received messages.
Second,
it cannot handle loops that implement complex state machines and all loops are assumed to process repetitive fields in a message.
\mreplace{In our analysis, we do not have this assumption and a parsing loop may process all fields in a message.}{\tool\ does not assume this.}
Rabkin and Katz~\cite{rabkin2011static} statically infer input formats in key-value forms, particularly for program configuration rather than networks.
\revise{Shoham et al.~\cite{shoham2008static} infer valid API sequences rather than message formats as state machines.}
Existing static analysis \mreplace{in the field of}{for} reverse engineering focuses on security protocols,
which, different from message formats, infers an agreed sequence of actions performed by multiple entities~\cite{avalle2014formal}.

\defparbf{Applications of Protocol Reverse Engineering.}
\mreplace{Revere engineering of m}{Formal m}essage formats \mreplace{is}{are} important for protocol fuzzing.
Mutation-based fuzzers use formats to generate the seed corpus~\cite{pham2020aflnet,gascon2015pulsar,gorbunov2010autofuzz,hu2018ganfuzz,chen2018iotfuzzer,somorovsky2016systematic}.
Generation-based fuzzers directly use the formats to generate messages for testing~\cite{boofuzz,peach,sulley,fuzzowski,banks2006snooze}.
Protocol model checking and verification also need formal protocol specifications~\cite{basin2018formal,benjamin2015messy,bhargavan2017verified,bishop2005rigorous,bishop2006engineering,blanchet2016modeling,cremers2018component,meier2013tamarin,musuvathi2004model,udrea2006rule}.
Blanchet~\cite{blanchet2016modeling} specifies a protocol by Horn clauses
and applies their technique to verify TLS models~\cite{bhargavan2017verified}.
Beurdouche et al.~\cite{benjamin2015messy} use Frama-C~\cite{kirchner2015frama}
to verify TLS implementations.
Tamarin~\cite{meier2013tamarin} uses a domain-specific language to establish proofs for security protocols
and applies to 5G AKA protocols~\cite{basin2018formal,cremers2018component}.
\mreplace{There are also a few}{Some} works \mreplace{on verifying}{verify} TCP components via symbolic analysis~\cite{bishop2005rigorous,bishop2006engineering,musuvathi2004model}.
Udrea et al.~\cite{udrea2006rule} use a rule-based static analysis to identify problems in protocols.
All these works assume the existence of formal specifications or manually build them.
We push forward the study of automatic specification inference and can infer message formats with high precision, recall, and speed.
\section{Conclusion}\label{sec:conclusion}

We present a \mdelete{practical }static analysis \mreplace{for inferring}{that infers} \mrevise{an FSM to represent} the format of regular protocols.
\mrevise{We significantly mitigate the path-explosion problem via carefully designed path merging and splitting rules.}
\mdelete{It is different from the state of the art\delete{s} in two aspects.
First, as static analysis, it does not rely on any network messages and overcomes the coverage problem of the mainstream dynamic-analysis-based approaches.
Second, 
instead of segmenting a message into multiple fields,
we lift protocol implementation to finite state machines, which provide a new perspective of understanding message formats in reverse engineering.}\mremark{remove to save space}Evaluation shows that our approach achieves high precision, recall, and speed.
Fuzzers supported by our work can achieve high coverage and discover zero-day bugs.
\newpage

%

\balance
\bibliographystyle{plain}
\bibliography{sigproc}

\newpage
\appendix

\startreviseblock

\section{\revise{Soundness and Completeness}}\label{app:soundness}

\setcounter{algocf}{0}
\myalg{\label{alg:copy}}

To facilitate the discussion and understanding of soundness and completeness,
we put a copy of our algorithm, i.e., Algorithm~\ref{alg:framework}, on top of this page.
The algorithm uses a worklist for fixed-point computation.
The worklist is a set of $(S, \bE_S)$ such that $S$ is a path set that we analyze in a loop iteration
and $\bE_S=(\bI_S, \phi_S)$ is the resulting program environment.
In the algorithm,
whenever we create a new $(S, \bE_S)$ (Line~14), or $(S, \bE_S)$ in the FSM does not reach the fixed point (Line~19),
we add it to the worklist.
Given each item $(S, \bE_S)$ popped from the worklist,
we create a state transition $(S, \bE_S, S')$.
Hence,
in what follows, we prove Theorem~\ref{th:sound} in three steps,
respectively
proving (1) the soundness/completeness of items in the worklist, i.e., $(S, \bE_S)$,
(2) the soundness/completeness of state transitions, i.e., $(S, \bE_S, S')$,
and (3) the soundness/completeness of the FSM, which is a set of transitions.

\begin{lemma}[Soundness of $(S, \bE_S)$]\label{lemma:sound1}
    For each variable $v$, $\bI_S(v)$ returns a sound abstract value that over-approximates all possible concrete values of the variable $v$.
\end{lemma}
\begin{proof}
    In Algorithm~\ref{alg:copy}, the pair $(S, \bE_S)$ in the worklist may
    come from three places:
    \ding{182} the ones produced by the abstract interpretation (Line~14 if we have $S'\equiv S_i'$, meaning that we actually do not split the state);
    \ding{183} the ones produced by splitting (Line~14);
    and \ding{184} the ones produced by merging (Line~19).
    Next, we explain that in each case, any abstract value $\bI_S(v)$ in the program environment is sound.

    \ding{182}
    Figure~\ref{fig:semantics_basic} shows a standard dataflow analysis for our abstract language model, i.e., Figure~\ref{fig:design_lang}.
    The analysis models the exact semantics of each program statement.
    For instance, if the abstract values of the variables $v_1$ and $v_2$ are respectively $v_1^\sharp$ and $v_2^\sharp$,
    the result of $v_1 \oplus v_2$ will be $v_1^\sharp \oplus v_2^\sharp$.
    Hence, each inference rule in Figure~\ref{fig:semantics_basic} is sound and complete.
    Given that each inference rule is sound and complete, 
    the analysis of each loop iteration is also sound and complete.
    Therefore,
    the resulting program environment is sound and complete,
    meaning that the abstract interpretation does not introduce any over- and under-approximation into the program environment.
    
    \ding{183}
    As shown in Algorithm~\ref{alg:splitting},
    when splitting a state $S$ to multiple states $S_i$,
    we rewrite each abstract value in $\bI_S$ via a simplification procedure to build $\bI_{S_i}$.
    This simplification procedure~\cite{dillig2010small} only rewrites a formula
    by removing abstract values from unreachable paths and, thus, does not introduce any over- and under-approximation into the program environment.
    For instance, assume $\bI_{S}(v) = \textup{ite}(\tilde{v}_1, \tilde{v}_2, \tilde{v}_3)$,
    meaning that after analyzing the path set $S$, the abstract value of the variable $v$
    is either $\tilde{v}_2$ or $\tilde{v}_3$, depending on if the branching condition $\tilde{v}_1$ is true.
    If paths in the subset $S_{i}\subseteq S$ ensures $\tilde{v}_1 = \textup{true}$,
    we then rewrite the abstract value as $\bI_{S_{i}}(v) =  \tilde{v}_2$.
    
    \ding{184}
    As shown in Algorithm~\ref{alg:mr2}, 
    when merging two program environments,
    we first convert them into intervals and then use the widening operator to merge them.
    Both the conversion and widening operations have been shown to be sound but not complete in literature~\cite{cousot1977abstract,yao2021program}.
    They are not complete because it introduces over-approximation into the abstract values.
    For instance,
    we may widen two intervals $[1, 3]$ and $[8, 10]$ to $[1, +\infty]$ which includes a large number of values, e.g., 5, not in the original intervals.
\end{proof}

In the worklist algorithm,
an FSM is a set of transitions, $(S, \bE_S, S')$,
which is actually $(S, \bE_S)$ together with the the path set $S'$ analyzed in the next loop iteration.
Intuitively,
if we have state transitions $(S, \bE_S, S'_1), (S, \bE_S, S'_2), (S, \bE_S, S'_3), \dots \in \textit{FSM}$,
it means that after executing a path $s\in S$ in a loop iteration,
we will execute a path $s'\in \bigcup S'_i$ in the next loop iteration. 
Next, we discuss the soundness of $(S, \bE_S, S')$ as follows.

\begin{lemma}[Soundness of $(S, \bE_S, S')$]\label{lemma:sound2}
    If in a concrete execution, two consecutive loop iterations respectively execute two paths in the loop body, e.g., $s$ and $s'$,
    there must exist a state transition $(S, \bE_S, S') \in \textit{FSM}$ such that $s\in S$ and $s'\in S'$.
\end{lemma}
\begin{proof}
        By Lemma~\ref{lemma:sound1}, 
        the output environment of analyzing the path set $s\in S$ is sound,
        meaning that each abstract value in $\bI_S$ over-approximates values in the concrete path $s$.
        Due to the over-approximation,
        using $\bE_S$ as the initial program environment,
        the next loop iteration must analyze a path set $S'$ that includes $s'$.
        If $S$ and $S'$ are not further split into sub-states in Algorithm~\ref{alg:copy},
        we have $(S, \bE_S, S')\in \textit{FSM}$. Hence, the lemma is proved.
        
        If $S$ and $S'$ are split into smaller sub-states, e.g., $s\in S_i$ and $s'\in S_i'$,
        Lines~9-11 in Algorithm~\ref{alg:copy} say that
        we still preserve the connections between $S_i$ and $S_i'$.
        Hence, we have $(S_i, \bE_{S_i}, S'_i)\in \textit{FSM}$. The lemma is also proved.
\end{proof}

Given that the transitions inferred by Algorithm~\ref{alg:copy} is sound, 
we discuss the soundness of the whole FSM below.

\begin{lemma}[Soundness of \textit{FSM}]\label{lemma:sound3}
    If a network message can be accepted by the loop under analysis, it can also be accepted by our inferred FSM.
\end{lemma}
\begin{proof}
        If the original program can accept an input message,
        then the input message will execute a sequence of paths, e.g., $(s_1, s_2, \dots, s_n)$,
        such that $s_i$ is a path in the loop body and is executed in the $i$th loop iteration, and $s_n$ is a path ending with an \textit{exit} statement.
        Assuming that the exact path constraint (i.e., path constraint without over- and under-approximation) of each path $s_i$ is $\Gamma_{s_i}$,
        we can write the exact path constraint of the whole input message as $\bigwedge_{i=1}^n\Gamma_{s_i}$.
        
        By Lemma~\ref{lemma:sound2},
        for each pair of path $(s_i, s_{i+1})$, we can find a state transition 
        $(S_i, \bE_{S_i}, S_{i+1})$ such that $s_i\in S_i$ and $s_{i+1}\in S_{i+1}$.
        By Lemma~\ref{lemma:sound1},
        $\bE_{S_i}$ is sound, meaning that the state transition from $S_i$ to $S_{i+1}$
        is constrained by a sound path constraint $\Phi_{S_i}$ such that $\Gamma_{s_i} \Rightarrow \Phi_{S_i}$.
        Therefore, $\bigwedge_{i=1}^n\Gamma_{s_i} \Rightarrow \bigwedge_{i=1}^n\Phi_{S_i}$.
        This means that the input message also satisfies $\bigwedge_{i=1}^n\Phi_{S_i}$.
        Thus, the state transitions from the state $S_1$ to the state $S_n$ can consume the whole input message.
        
        Finally, due to SR2, $S_n$ is a final state. Hence, our inferred FSM also accepts the input message.
\end{proof}

The completeness of our inferred FSMs can be discussed in a similar manner as below.

\begin{lemma}[Completeness of \textit{FSM}]\label{lemma:completeness}
    Assuming we do not use any interval domain during our analysis,
    the inferred FSM is complete --- 
    if a message can be accepted by our inferred FSM,
    it can also be accepted by the loop under analysis.
\end{lemma}
\begin{proof}
    As discussed in the proof of Lemma~\ref{lemma:sound1},
    we only introduce over-approximation into the program environment in the third case when the interval domain is used.
    Hence, $(S ,\bE_S)$ is complete if the interval domain is never used.
    In this case, each state transition in the FSM, i.e., $(S ,\bE_S, S')$, is constrained by the exact path constraint.
    That is, we have
    $\forall s\in S, \Phi_{s} \Leftrightarrow \Gamma_{s}$ where $\Phi_{s}$ and $\Gamma_{s}$ respectively denote the inferred and the exact path constraints of the path $s$.
    The transition constraint is denoted by $\Phi_{S} = \bigvee_{s\in S} \Phi_{s}$.
    
    If our inferred FSM can accept a message,
    then there is a sequence of state transitions $(S_1, S_2, \dots, S_n)$ that can consume the message.
    That is, the message satisfies $\bigwedge \Phi_{S_i}$, i.e.,
    $$
    \bigvee_{s_1\in S_1} \Phi_{s_1} \land \bigvee_{s_2\in S_2} \Phi_{s_2} \land \bigvee_{s_3\in S_3} \Phi_{s_3} \land \dots \land \bigvee_{s_n\in S_n} \Phi_{s_n}.
    $$
    
    We can then pick one path $s_i$ from each path set $S_i$ such that 
    the network message satisfies $\bigwedge \Phi_{s_i}$.
    Since $\Phi_{s} \Leftrightarrow \Gamma_{s}$ as discussed before,
    the network message also satisfies $\bigwedge \Gamma_{s_i}$.
    This means
    the loop under analysis can consume the network message via the path sequence $(s_1, s_2, \dots, s_n)$.
    
    Finally, since $S_n$ is a final state, SR2 ensures that $s_n\in S_n$ ends with a loop-exiting statement, meaning that
    the loop under analysis accepts the network message.
\end{proof}

\section{\revise{Computing the Precision and Recall of FSM}}\label{app:compare}

In order to identify correct or incorrect transitions in an inferred FSM, which is necessary for calculating the precision and recall,
we cannot directly use a graph differentiating algorithm to compare the ground-truth FSM with the inferred FSM, due to the following reasons. 
First, multiple FSMs, whether equivalent or not, may be represented in completely different graph structures.
Thus, given a transition in the inferred FSM, it could be hard to find its correct counterpart in the ground truth, thereby being hard to determine the correctness of the transition.
Second, it makes little sense to discuss the correctness of a single transition,
because a state transition could be correct for parsing one message but incorrect for another.
For instance, assume the FSM in Figure~\ref{fig:b2}(a) is the ground truth, which supports two message types: the first can be parsed from $A$ to $D$ and the second from $A$ to $D'$.
Figure~\ref{fig:b2}(b) is an inferred FSM, where the transition from $C$ to $D$ is correct for the first message type but incorrect for the second.
Thus, when computing the precision and recall,
we need to consider the correctness of a state transition in the context of a full path from the start state to the final state. 
In what follows, we outline an approach to estimating the precision and recall of an inferred FSM.

\begin{figure}[h]
    \centering
    \includegraphics[width=0.75\linewidth]{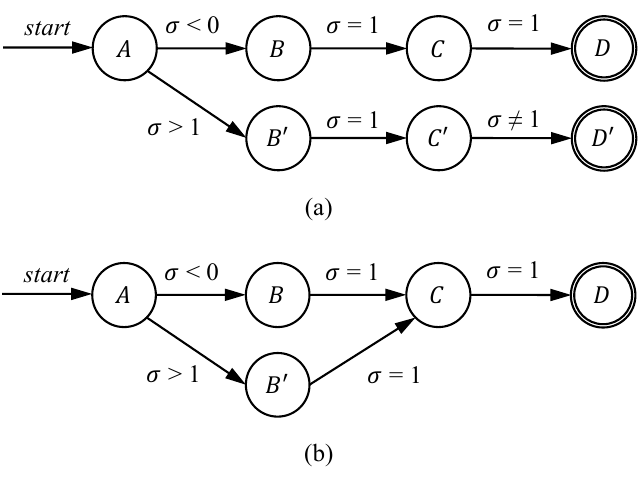}
    \vspace{-4mm}
    \caption{Correctness of a single transition.}\label{fig:b2}
    \vspace{-2mm}
\end{figure}

\noindent
\textbf{Step 1: Extracting Formats from Official Documents.}
For each protocol, we manually build a formal format as a ce-regex based on its official document.
This format serves as the ground truth in our evaluation.
Note that manually building a format based on the official document
is a common practice and frequently used in almost all literature on protocol reverse engineering, such as Tupni~\cite{lin2008automatic} and AutoFormat~\cite{cui2008tupni}.

For instance, for the protocol Mavlink, a snippet of the manually-built format is as below.
\begin{center}\small
    \texttt{STX(1) $\cdots$ Sys-ID(1) Comp-ID(1) Msg-ID(3) $\cdots$}
\end{center}
It indicates that a Mavlink message is a byte sequence that can be split into multiple fields,
including \texttt{STX}, \texttt{Sys-ID}, \texttt{Comp-ID}, and \texttt{Msg-ID}.
The first three fields are one-byte integers and the field \texttt{Msg-ID} is a three-byte integer.
As a ce-regex, the manually built format also includes constraints like \texttt{STX = 0xFD},
which says that the first field must be a constant 0xFD.
It is direct to transform the ce-regex to an FSM.

\smallskip
\noindent
\textbf{Step 2: Normalizing an FSM.}
We normalize the FSMs so that we can evaluate the correctness of state transitions at a fine-grained level.
If a state transition from the state $A$ to the state $B$, e.g., $\delta(A, \alpha\lor \beta) = \{B\}$, is constrained by a disjunctive constraint, e.g., $\alpha\lor \beta$,
we split it into two state transitions, i.e., $\delta(A, \alpha) = \{B\}$ and $\delta(A, \beta) = \{B\}$,
which are respectively constrained by $\alpha$ and $\beta$.
Assume the inferred constraint $\alpha$ is incorrect but $\beta$ is correct.
Before normalization, since $\alpha$ is incorrect, the constraint $\alpha\lor \beta$ is regarded to be incorrect.
As a result, the state transition is also considered incorrect.
After normalization,
we can evaluate the state transition at a fine-grained level.
That is, the transition with the constraint $\alpha$ is incorrect but the one with $\beta$ is considered correct.
This normalization rule is illustrated in Figure~\ref{fig:norm}(a).

\begin{figure}[h]
    \centering
    \includegraphics[width=0.8\linewidth]{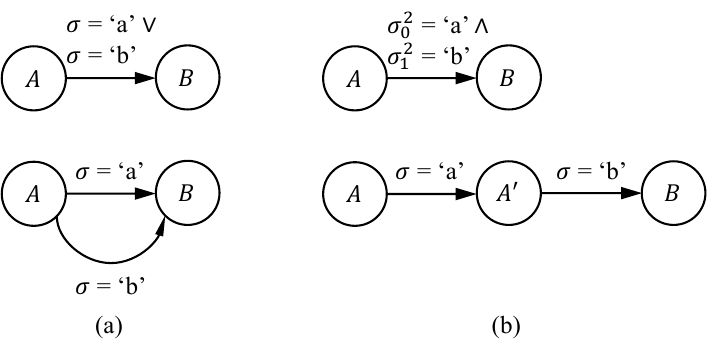}
    \vspace{-4mm}
    \caption{Normalization.}\label{fig:norm}
\end{figure}

Similarly, if a state transition $\delta(A, \alpha\land \beta) = \{B\}$ is constrained by
a conjunctive constraint where the constraints $\alpha$ and $\beta$ respectively constrain two independent inputs, e.g.,
$\sigma^{n}_{0..i}$ and $\sigma^{n}_{i+1..n-1}$,
we split the transition into two consecutive transitions, $\delta(A, \alpha) = \{A'\}$ and $\delta(A', \beta) = \{B\}$.
This normalization rule is illustrated in Figure~\ref{fig:norm}(b).

\smallskip
\noindent
\textbf{Step 3: Computing Precision and Recall.}
As discussed before,
we should consider the correctness of a transition in the context of a full path in an FSM.
However, due to path explosion, sometimes,
we cannot enumerate all paths in an FSM.
Instead, we enumerate all paths of length 1, 2, 3, $\cdots$, in the ground-truth FSM, until either we get 1 million paths (we believe the number of paths is sufficiently large) or we have enumerated all paths in the FSM.
We record the path set as $P$.

For each path $p\in P$ in the ground truth, to find its corresponding path $p'$ in the inferred FSM, we generate a message by solving its path constraint and use the inferred FSM to parse the message.
We then compare the two paths $p$ and $p'$.
A transition in $p'$ is correct iff it has the same constraint as $p$.
The number of correct transitions is denoted as $T(p')$ and incorrect transitions $F(p')$.
We then respectively compute the precision and recall of the inferred FSM as follows.
$$
\textit{Precision} = \frac{\Sigma_{p\in P} T(p')}{\Sigma_{p\in P} T(p') + F(p')};
\enskip\enskip\enskip
\textit{Recall} = \frac{\Sigma_{p\in P} T(p')}{\Sigma_{p\in P} T(p)}
$$

\noindent
Intuitively,
the precision is the ratio of correct state transitions to all inferred transitions;
and the recall is the ratio of correct state transitions to all transitions in the ground truth.

\dismissreviseblock

\mstartreviseblock

\section{Example of Detected Bugs}\label{app:bugs}

\mremark{The example is originally put in Section 6. We move it to the appendix to save space.}
Figure~\ref{fig:bug} shows a global buffer overflow detected in SML.
The parsing loop calls the function \textit{smlState} with an input byte
and updates the state according to the byte.
In each parsing iteration, it pushes at most two bytes into the global buffer \textit{listBuffer}, of which the maximum length is 80.
This means that to trigger the bug in the function \textit{smlOBISManufacturer}, a message has to pass at least 40 iterations of the parsing loop.
In other words, the bug-triggering message must pass at least 40 state transitions in an inferred FSM.
This requires the inferred FSM to be of high precision and recall. Otherwise, ill-formed messages will be generated,
which are easy to be pruned and cannot execute deep program paths.
As discussed in \S\ref{sec:motivating_example}, conventional approaches cannot handle FSM parsing loops well. Thus, fuzzers armed with them miss this bug.

\begin{figure}[h]
    \centering
    \includegraphics[width=0.95\columnwidth]{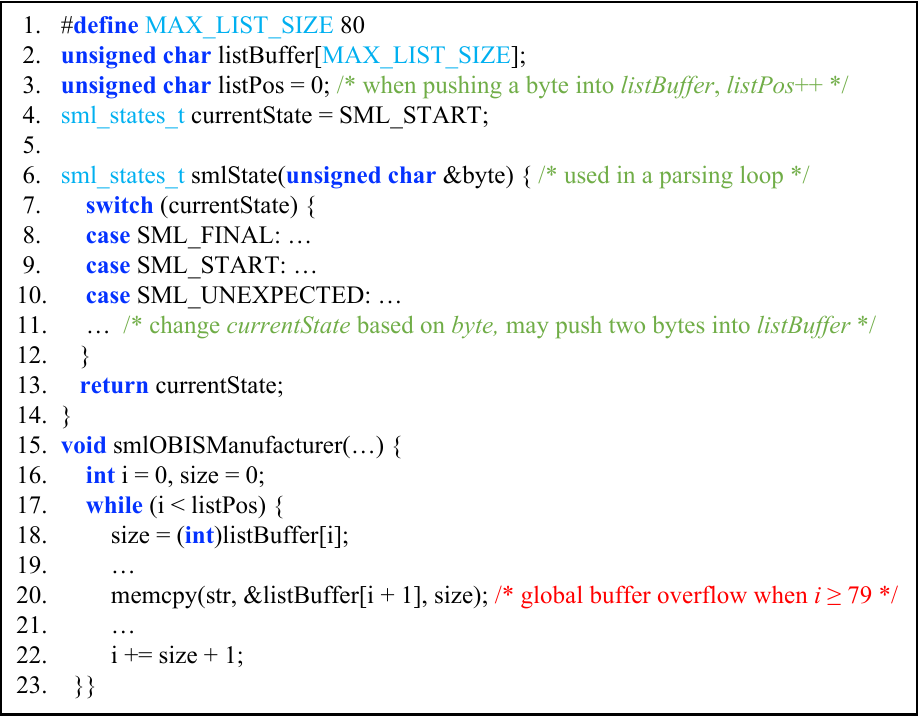}
    \vspace{-2mm}
    \caption{An example of detected vulnerabilities.}\label{fig:bug}
    \vspace{-4mm}
\end{figure}

\mdismissreviseblock

\end{document}